\documentclass[10pt]{article}
\usepackage{graphicx}
\usepackage{amsmath,amsthm}
\usepackage{tikz}
\usepackage{longtable}

\newtheorem{theorem}{Theorem}
\newtheorem{lemma}{Lemma}

\begin{document}
\title{\sf  Bollinger Bands Thirty Years Later }
%% I TOOK THIS OUT ALSO
%%\author{Mark Leeds\symbolfootnote[1]{This is a footnote}}
%%\author{Mark Leeds}\footnote[2]{This is a footnote}
\author{Mark Leeds}
\date{{\today}} 
%%\date{\small\it September 22, 2011} 
%%%% THIS IS BETTER. If you use the above, whenever you complie, the date will change per each different day.  

\maketitle

%%========================================================================================================================================================================================

\vspace{-5ex}
\begin{abstract}
\noindent The goal of this study is to explain and examine the statistical underpinnings of the Bollinger Band methodology. We start off by
elucidating the rolling regression time series model and deriving its explicit relationship to Bollinger Bands. Next we illustrate the
use of Bollinger Bands in pairs trading and prove the existence of a specific return duration relationship in Bollinger
Band pairs trading\cite{INV2007}. Then by viewing the Bollinger Band moving average as an approximation to the random walk plus noise (RWPN) 
time series model, we develop a pairs trading variant that we call ``Fixed Forecast Maximum Duration' Bands'' (FFMDPT). Lastly, we conduct 
pairs trading simulations using SAP and Nikkei index data in order to compare the performance of the variant with Bollinger Bands.

\end{abstract}

%%========================================================================================================================================================================================

\clearpage
\section{Introduction}\label{S:bbdesc}
Developed by John Bollinger in the early 1980's, the Bollinger Band methodology is a frequently used tool in the analysis of financial markets. 
Traders frequently use the outputs of Bollinger Bands in conjunction with other technical indicators in order to choose the position to take in the asset being monitored. Although Bollinger Bands are a common tool for analyzing asset behavior, the
Bollinger Band components have generally been viewed as outputs of an algorithm rather than as estimates of the parameters of a statistical model. 
More details about the history and development of Bollinger Bands can be found in \cite{JB01} and \cite{JB02}. 

\noindent A basic explanation of the Bollinger Band construction follows. Given a time series $y_{t}$ at $t=t^{*}$, define the n day rolling moving average of the series as $mave_{t^{*}}$:
\begin{equation}\label{E:bbmave}
mave_{t^{*}} = \sum_{t= t^{*} - n+1}^{t=t^{*}} y_{t}/n  ~~~~~ t^{*} = n,\ldots, T
\end{equation}

\noindent Note that, because the moving average uses $n$ data points, the first time $t^{*}$ at which the $mave_{t^{*}}$ can be calculated is at 
$t^{*} = n$.
\noindent Similarly, the $n$ day rolling variance at time $t = t^{*},~\sigma^2_{t^{*}}$, is defined as:
\begin{equation}\label{E:bbvar}
\hat{\sigma}^2_{t^{*}} \footnote{Note that this is the formula used to obtain an unbiased estimate of the unknown variance, $\sigma^2_{t^{*}}$. Another common way of defining the 
variance, $\hat{\sigma^2}_{t^{*}}$, is to use $n$ in the denominator rather than $(n-1)$. The results that follow are dependent on the denominator in equation~(\ref{E:bbvar}) being defined  
as (n-1).} = \sum_{t= t^{*} - n+1}^{t=t^{*}}(y_{t} - mave_{t^{*}})^{2}/(n-1)  ~~~~~ t^{*} = n,\ldots, T
\end{equation}
\noindent 
Then, given the relations above, the Bollinger Band components are constructed using a center line and an upper and lower band defined 
respectively as:
\begin{equation*}
CL_{t^{*}} = mave_{t^{*}}  
\end{equation*}
\begin{equation*}
BBupper_{t^{*}} = mave_{t^{*}} + k * \hat{\sigma}_{t^{*}}
\end{equation*}
\noindent and
\begin{equation*}
BBlower_{t^{*}} = mave_{t^{*}} - k * \hat{\sigma}_{t^{*}}
\end{equation*}
\noindent where $k$ is referred to as the width multiplier and represents the distance in standard deviation units from the center line to 
each band. An example of the Bollinger Band construction is shown in Figure~\ref{FIG:boll_pic1} in Appendix A on page \pageref{FIG:boll_pic1}.

\newpage

\noindent The use of the moving average for the center line has generally been viewed by market technicians as a low pass filter for the time series being monitored. By calculating 
the moving average of the actual series and plotting the resulting series, the high frequency component is eliminated from the original series and only the trend remains. The upper and lower band 
calculations use this trend as an input and they are useful for developing indicator rules such as ``when the price crosses $BBUpper$ and the RSI is above X, this indicates that
the price is expected to \ldots ``. As far as the the origin of the dispersion component is concerned, many different types of bands were experimented with before John Bollinger 
came up with the idea of using the sample standard deviation, $\hat\sigma$, as the measure of the current dispersion of the time series. The details of his discovery are captured quite 
vividly in \cite{JB01} and are left for the reader to explore. 

\noindent The goal of this study is to make connections between Bollinger Bands and time series models and show how these connections 
can lead to useful statistical insights. The first connection shows that although Bollinger Bands are generally viewed as a somewhat
ad-hoc algorithm that generate outputs used as trading indicators, they actually  have strong statistical foundations. The second
connection provides an alternative way of viewing the Bollinger Band pairs trading algorithm and leads to an interesting Bollinger 
Band variant. 

\section{Bollinger Band Literature}

\noindent The literature with respect to Bollinger Bands simulations is quite vast. Butler and Kazakov \cite{BK10} apply swarm optimization techniques to search for 
optimal Bollinger Band Bollinger parameters. The optimizations are done with respect to the profit and loss of Bollinger Band pairs trading strategies.\footnote{The 
application of Bollinger Bands to pairs trading will be discussed in detail in Section~\ref{S:BBPT}.} Similarly, Ni and Zhang \cite{NZ2007} use genetic algorithms 
to find the optimal Bollinger Band window length and band width jointly.  The research regarding variations on Bollinger Bands is less 
plentiful. Oleksiv  \cite{OL10} uses different 
algorithms for the construction of the bands including kriging, a method more common in geostatistics. Chande \cite{TSC92} uses an exponentially weighted moving 
average as a low pass filter for prices and adjusts the smoothing parameter dynamically based on the volatility of prices. Finally, Tilley \cite{DLT98} combines the 
moving average with the concept of support and resistance in order to switch between emerging markets funds and small cap funds to and from the SAP 500.

\noindent The rest of this article is organized as follows. In Section \ref{S:RR} we demonstrate an equivalence between Bollinger Bands and the rolling regression time series model. 
In Section \ref{S:BBPT} we describe how Bollinger Bands can be used in pairs trading as a mechanism for capturing the mean reversion behavior expected in the asset 
pair being traded. In Section \ref{S:FFMDPT} , we make a connection between Bollinger Bands and a state space model called the random walk plus noise model. This 
connection provides another approximate statistical framework for Bollinger Bands and leads to a variant of Bollinger Bands called Fixed Forecast Maximum Duration 
Bands. We then construct a pairs trading simulation in order to compare the out of sample performance of the Bollinger Bands pairs trading strategy (BBPT) and the 
Fixed Forecast Maximum Duration pairs trading strategy (FFMDPT). Finally, in Section \ref{S:FUTRES}, we summarize our findings and provide suggestions for future 
research areas.

%%=========================================================================================================================================================

\section{Bollinger Bands as a Rolling Regression Time Series Model}\label{S:RR}
 In order to develop a connection between Bollinger Bands and the rolling regression time series model, we first need to describe the latter in precise 
detail.\footnote{The originator of the rolling regression model is not known by the author but its popularity is most likely due to Fama and MacBeth \cite{FM73}.}

\subsection{The Rolling Time Series Regression Model}\label{SS:RR}
\noindent  The rolling regression time series model is commonly used when model coefficients are expected to change over time. 
Following the notation of Zivot and Wang \cite{ZW03}, the rolling regression time series model using an n day moving window is shown below:
\begin{equation} \label{E:rollingreg}
\textbf{y}_{t^{*}}(n) = \textbf{X}_{t^{*}}(n) \boldsymbol{\beta}_{t^{*}}(n)  + \boldsymbol{\epsilon}_{t^{*}}(n)  ~~ t^{*} = n,\cdots T
\end{equation}
\noindent Here $\textbf{y}_{t^{*}}(n)$ is an ($n \times 1$) vector of independent observations on the response, $\textbf{X}_{t^{*}}(n)$ is an ($n \times k$) matrix of explanatory 
variables and  finally $\boldsymbol{\epsilon}_{t^{*}}(n)$ is an ($n \times 1$) vector of error terms each being $\sim N(0,{\sigma^{2}}_{t^{*}})$ . Note that $(n)$ indicates that the the $n$ 
observations in $\textbf{y}_{t^{*}}(n)$ and $\textbf{X}_{t^{*}}(n)$  are the n most recent values from time ($t^{*} - n + 1$) to $t^{*}$. Clearly we need to assume that $n > k$. 

\noindent  It is important to understand what is being assumed by the use of the $(n)$ notation. First of all, although the new observation at time $t^{*}$ is univariate, 
at time $t^{*}$, the vector $\textbf{y}_{t^{*}}(n)$ of observations from $t^{*}-n+1$ to $t^{*}$ is used to estimate $\boldsymbol{\beta}_{t^{*}}(n)$. Therefore, we need to differentiate between the new univariate 
observation at $t^{*}$ and the n-dimensional vector of observations at $t^{*}$, $y_{t^{*}}(n)$. In what follows, we always refer to the  
$n \times 1$ vector at some $t = t^{*}$ 
as $vecobs_{t^{*}}$ and the new univariate observation seen at t = $t^{*}$  as $uniobs_{t^{*}}$. 

\noindent The rolling regression estimation algorithm proceeds in the following manner: Initially, we start out at $t^{*} = n$ because that is first point at 
which we can construct an estimate of $\boldsymbol{\beta}$. We observe $vecobs_{t^{*} = n}$ which is the n-dimensional vector of the first n observations in the series. Note that $vecobs_{t^{*}}$ has a 
regression model associated with it, namely, $vecobs_{t^{*}} = \textbf{X}_{t^{*}}\boldsymbol{\beta}_{t^{*}} + \boldsymbol{\epsilon}_{t^{*}}$ with the  error term $\boldsymbol{\epsilon}_{t^{*}}$ 
assumed to be independent (i.e. zeros off the diagonal of its covariance matrix). So, $vecobs_{t^{*}}$ is observed and the coefficients, $\boldsymbol{\beta}_{t^{*}}$, in the model are
then estimated. Next, time proceeds from $t = t^{*}$ to $t = t^{*} + 1 = n + 1$ and a new observation $vecobs_{t^{*}+1}$ is observed. But this supposedly new observation is constructed in the 
following manner: $uniobs_{t^{*}-(n-1)}$ is removed from $vecobs_{t^{*}}$ and the new $uniobs_{t^{*}+1}$ is observed and added to the front of the $vecobs_{t^{*}}$ observation. This modified 
$vecobs_{t^{*}}$ vector is now $vecobs_{t^{*}+1}$ and is the ``new'' n-dimensional observation at $t=t^{*}+1$. Again, $vecobs_{t^{*}+1}$ has a regression model associated with it namely, 
$y_{t^{*}+1} = \textbf{X}_{t^{*}+1}\boldsymbol{\beta}_{t^{*}+1} + \boldsymbol{\epsilon}_{t^{*}+1}$. The error term $\boldsymbol{\epsilon}_{t^{*}+1}$ is again assumed to be independent. So, 
once $vecobs_{t^{*}+1}$ is observed, the coefficients in the associated regression model are estimated thereby obtaining a new set of $\boldsymbol{\beta}_{t}$ coefficients at time $t = t^{*}+1$. 
This process repeats itself again at $t^{*} = n + 2$ and, $n + 3,\cdots$ and so on and so forth until $t^{*} = T$.

\noindent Note that there is a serious statistical problem with the model represented in equation (\ref{E:rollingreg}). 
Clearly the response $vecobs_{t^{*}}$  is highly correlated  with the response $vecobs_{t^{*}+1}$ because of how these observations are constructed. In fact, any two n-dimensional observations 
$vecobs_{t^{\prime}}$ and  $vecobs_{t^{\prime\prime}}$ constructed less than n periods apart will be correlated because they will contain common observations due to the rolling window construction. 
Now, even though this correlation exists, the rolling regression methodology still assumes that each regression model has independent error terms and therefore independent 
$vecobs_{t^{*}} ~~\forall t^{*}=n.\cdots, T$. We should define the assumption more rigorously. Formally, let us assume that 
the probability at time t of $vecobs_{t}$ ~\text(i.e. $\textbf{Y}_{t}$) possesses the following property:
\begin{equation}
\mbox{Prob}( \textbf{Y}_{t^{*}} = \textbf{y}_{t^{*}}) | \, \textbf{B}_{t^{*}}) = \mbox{Prob}( \textbf{Y}_{t^{*}} = \textbf{y}_{t^{*}}) 
\end{equation}
\noindent where 
$$
\textbf{B}_{t^{*}} = \{\, \textbf{Y}_{t} \,\,\,t = 1,\cdots, t^{*}-1 \,\} 
$$
\noindent This assumption implies that the likelihood of any $vecobs_{t^{*}}$ is independent of the previous $vecobs_{t}$ observations even though this is clearly not true.
Why is this assumption required ? Often it is believed that, 
due to structural changes or simply noise , the $\boldsymbol{\beta}_{t^{*}}$ parameter is expected to change over time. Yet, at the same time, one also knows with certainty that  
estimates of $\boldsymbol{\beta}_{t^{*}}$ that are close to each other in time are highly positively correlated. Therefore, the only way to generate correlation in the estimates, 
allow them to change over time and yet keep the model analytically tractable without resorting to more complex techniques is to make this independence assumption. Rather
than imposing a model for $\boldsymbol{\beta}_{t^{*}}$ and allowing the data to speak for the new estimate of $\boldsymbol{\beta}_{t^{*}}$, each time there is a new data point,
the assumption is that, at each time t, a totally new n-dimensional data point is observed. In essence, from a time series modelling standpoint, the rolling window construction together with the 
independence assumption is an ad-hoc way of dealing with the fact that dynamics are not being specified for $\boldsymbol{\beta}_{t^{*}}$. The expected correlation of the $\boldsymbol{\beta}_{t^{*}}$ 
estimates is achieved through the use of the constant overlap in the adjacent n-dimensional observations. Intuitively, a larger window will generate more highly correlated estimates than a shorter 
window. The statistical flaw of the rolling regression time series model is that the independence assumption clearly does not hold so the estimates are biased with respect to the true underlying DGP. 

\noindent Conversely, the well known time varying regression-Kalman filter type model, also quite popular in econometrics, is more complex than the rolling regression model mathematically 
but has the advantage that only one model is assumed from the start and the dynamics for the beta coefficients are specified directly. Consequently there is no need for the ad-hoc 
construction of a rolling window. In the Kalman kilter framework, when a new $uniobs_{t^{*}}$  is observed at a new time $t=t^{*}$, the current model estimate, $\beta_{t^{*}-1}$, is updated and 
becomes the new estimate at $t^{*}$. This is probably why the rolling regression time series model is often referred to as the  ``poor man's time varying coefficient regression model''. 
More details on the Kalman filtering approach can be found in \cite{AJ70} and \cite{AH92} and it will also be discussed in more detail in Section \ref{SS:KF}.

\noindent Below, Figure~\ref{FIG:windowreg1} displays the relationship between adjacent windows in the rolling time series regression model at $t = t^{*}$ (red line segment) and $t = t^{*} + 1$.
(green line segment). Figure~\ref{FIG:windowreg2} displays what is assumed to be happening with the same adjacent windows.

%%===========================================================================================================================================
%% FIRST LINE FIGURE
%%===========================================================================================================================================

\vspace{0.3in}

\begin{figure}[htb]
\begin{center}
\caption{The rolling regression windows at ${t = t^{*}}$ and ${t = t^{*} + 1}$ contain common observations.}
\label{FIG:windowreg1}

$t = t^{*}$

\begin{tikzpicture}[xscale=8] 
\draw[-][draw=red, very thick] (0,0) -- (1.5,0);
\draw [thick] (0,-.1) node[below]{$y_{t^{*} - n + 1}$} -- (0,0.1);
\draw [thick] (0.25,-.1) node[below]{ . . . . . . . . . }  -- (0.25,0.1);
\draw [thick] (0.50,-.1) node[below]{ . . . . . . . . .} -- (0.50,0.1);
\draw [thick] (0.75,-.1) node[below]{$y_{t^{*}- n + 10}$} -- (0.75,0.1);
\draw [thick] (1.0,-.1) node[below]{. . . . . . . . . }  -- (1.0,0.1);
\draw [thick] (1.25,-.1) node[below]{$y_{t^{*}- n + (n-1)}$} -- (1.25,0.1);
\draw [thick] (1.5,-.1) node[below]{$y_{t^{*}}$} -- (1.50,0.1);
\end{tikzpicture}

\vspace{0.25in}

$t = t^{*} + 1$

\begin{tikzpicture}[xscale=8]
\draw[-][draw=green, very thick] (0,0) -- (1.5,0);
\draw [thick] (0,-.1) node[below]{$y_{t^{*}- n + 2}$} -- (0,0.1);
\draw [thick] (0.25,-.1) node[below]{. . . . . . . . . . . . }  -- (0.25,0.1);
\draw [thick] (0.50,-.1) node[below]{. . . . . . . . . . . . } -- (0.50,0.1);
\draw [thick] (0.75,-.1) node[below]{$y_{t^{*}- n + 11}$} -- (0.75,0.1);
\draw [thick] (1.0,-.1) node[below]{. . . . . . . . . . . }  -- (1.0,0.1);
\draw [thick] (1.25,-.1) node[below]{$y_{t^{*}}$}  -- (1.25,0.1);
\draw [thick] (1.5,-.1) node[below]{$y_{t^{*}+1}$} -- (1.50,0.1);
\end{tikzpicture}

\end{center}
\end{figure}
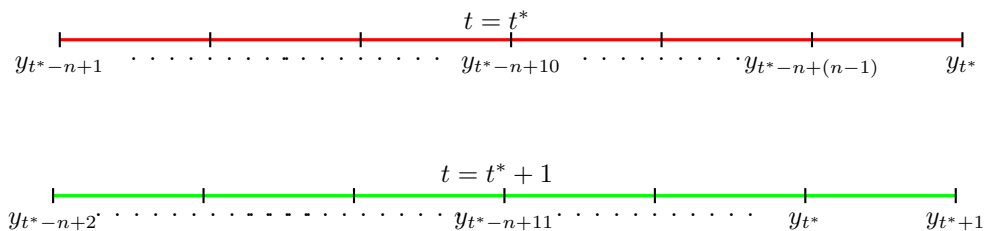

%%======================================================================================================================
%% SECOND LINE FIGURE
%%======================================================================================================================

\vspace{0.2in}

\begin{figure}[htb]
\begin{center}
\caption{Although Figure~\ref{FIG:windowreg1} on page \pageref{FIG:windowreg1} clearly shows the contrary, the assumption in the rolling regression model is that adjacent windows at ${t = t^{*}}$ and ${t = t^{*}+1}$ do not 
contain common  observations.}

$t = t^{*}$

\begin{tikzpicture}[xscale=8] 
\draw[-][draw=red, very thick] (0,0) -- (1.5,0);
\draw [thick] (0,-.1) node[below]{$y^{\prime}_{t^{*}- n + 1}$} -- (0,0.1);
\draw [thick] (0.25,-.1) node[below]{  . . . . . }  -- (0.25,0.1);
\draw [thick] (0.50,-.1) node[below]{. . . . . . . . . .} -- (0.50,0.1);
\draw [thick] (0.75,-.1) node[below]{$y^{\prime}_{t^{*} -n + 10}$} -- (0.75,0.1);
\draw [thick] (1.0,-.1) node[below]{. . . . . . . . . }  -- (1.0,0.1);
\draw [thick] (1.25,-.1) node[below]{$y^{\prime}_{t^{*}- n + (n-1)}$} -- (1.25,0.1);
\draw [thick] (1.5,-.1) node[below]{$y^{\prime}_{t^{*}}$} -- (1.50,0.1);
\end{tikzpicture}

\vspace{0.25in}

$t = t^{*} + 1$

\begin{tikzpicture}[xscale=8]
\draw[-][draw=green, very thick] (0,0) -- (1.5,0);
\draw [thick] (0,-.1) node[below]{$y^{\prime\prime}_{t^{*}- n + 2}$} -- (0,0.1);
\draw [thick] (0.25,-.1) node[below]{  . . . . . }  -- (0.25,0.1);
\draw [thick] (0.50,-.1) node[below]{. . . . . . . . . .} -- (0.50,0.1);
\draw [thick] (0.75,-.1) node[below]{$y^{\prime\prime}_{t^{*}- n + 11}$} -- (0.75,0.1);
\draw [thick] (1.0,-.1) node[below]{. . . . . . . . . }  -- (1.0,0.1);
\draw [thick] (1.25,-.1) node[below]{$y^{\prime\prime}_{t^{*}}$} -- (1.25,0.1);
\draw [thick] (1.5,-.1) node[below]{$y^{\prime\prime}_{t^{*}+1}$} -- (1.50,0.1);
\end{tikzpicture}

\label{FIG:windowreg2}
\end{center}
\end{figure}

%% END OF SECOND LINE FIGURE
%%============================================================================================================================

\newpage

\subsection{The Equivalence of Bollinger Bands and the Rolling Regression Time Series Model} \label{SS:RRBB}
\noindent Let us consider the intercept only version of the  rolling regression time series model represented by (\ref{E:rollingreg}) so that $\textbf{X}_{t^{*}}(n)$ is a vector of 
ones and  $\boldsymbol{\beta}_{t^{*}}(n)$ is a scalar. The rolling regression model in (\ref{E:rollingreg}) then becomes:
\begin{equation} \label{E:rollingreg2}
\textbf{y}_{t^{*}}(n) =  \beta_{t^{*}}(n) \textbf{X}_{t^{*}}(n) + \boldsymbol{\epsilon}_{t^{*}}(n)  ~~ t^{*} = n,\cdots T
\end{equation}

\noindent Using classical least squares results, it is straightforward to show that the estimates of this model at each time $t^{*}$ are:
\begin{equation}\label{E:rollingbetaest}
\hat{\beta}_{t^{*}} = \sum_{t=t^{*}-n+1}^{t=t^{*}}y_{t}/n  ~~~~~ t^{*} = n,\ldots, T
\end{equation}
\noindent and
\begin{equation}\label{E:rollingvarest}
\hat{\sigma}^{2}_{t^{*}}  = \sum_{t=t^{*}-n+1}^{t=t^{*}}(y_{t} - \hat\beta_{t^{*}}^{2})/(n-1) ~~~~~ t^{*} = n,\ldots, T
\end{equation}

\noindent Note that if we equate $\hat{\beta}_{t^{*}}$ and $mave_{t^{*}}$, then the expressions for the estimates in equations (\ref{E:rollingbetaest}) and (\ref{E:rollingvarest}) 
are exactly the same as the Bollinger Band components in equations~(\ref{E:bbmave}) and (\ref{E:bbvar})\footnote {In fact, the R \cite{R11} command:  rollapply(inseries, width = 
ndays, FUN = function(y) summary(lm( y $\sim$ 1))\$sigma, align = ``right'', fill = TRUE) and the R command: sqrt(rollapply(inseries, ndays, var, align = ``right'', fill = TRUE)) 
will give identical results given the same series  ``inseries''.}. Therefore, the Bollinger Band algorithm results in estimates that are identical to those of an intercept only 
rolling regression model where the intercept is the center line and the residual standard deviation is the Bollinger Band standard deviation.

\noindent It is also straightforward to show that, for the model in (\ref{E:rollingreg2}), a $100(1-\alpha)\%$ confidence interval for the future one step ahead response 
, also referred to in the statistics literature as a prediction interval \cite{JF2005}, is the following:
\begin{equation*}
\hat\beta_{t^{*}} ~~\pm~~t_{n-1}^{\alpha/2} \times \hat{\sigma}_{t^{*}}\sqrt{(1 + (1/n))}
\end{equation*}
\noindent Note that it is possible to approximate this $100(1-\alpha)\%$ prediction interval in the following way. For window values of n between 10 and 50, 
$\sqrt{(1 + 1/n)}$ ranges between 1.05 and 1.01. Therefore, for values of n between 10 and 50, the approximate $100(1-\alpha)\%$ prediction interval becomes: 
\begin{equation} \label{EQ:CI}
\hat\beta_{t^{*}} ~~\pm~~t_{n-1}^{\alpha/2} \times \hat{\sigma}_{t^{*}}
\end{equation}

\noindent But notice that, if $\hat\beta_{t^{*}}$ is replaced by $mave_{t^{*}}$ and $t_{n-1}^{\alpha/2}$ is replaced by $k$ , then (\ref{EQ:CI}) reduces to
\begin{equation*}
BBupper_{t^{*}} = mave_{t^{*}} + k * \hat{\sigma}_{t^{*}}
\end{equation*}
\noindent and
\begin{equation*}
BBlower_{t^{*}} = mave_{t^{*}} - k * \hat{\sigma}_{t^{*}}
\end{equation*}
Therefore, assuming that k is chosen appropriately, the upper and lower bands in Bollinger Bands are approximately equivalent to the prediction intervals 
constructed from an intercept only rolling regression model. We should point out that the approximation does depend on the order of magnitude of $\sigma$ and will get worse as $\sigma$ 
increases. At the same time,  it is always possible to avoid the approximation by including the 1/n factor in the construction of BBupper and BBlower and obtain the exact prediction interval 
associated with the intercept only rolling regression time series model. 

\noindent In summary, the Bollinger Band methodology can be viewed as an intercept only rolling regression time series model with the center line and the standard deviation being the mean and 
residual standard deviation from the rolling regression model respectively. The $BBUpper$ and $BBLower$ components of the Bollinger Bands will be approximately the same 
as the prediction intervals of the intercept only rolling regression model as long as k is chosen appropriately. 
\vspace {0.1in}

\newpage
%%======================================================================================================================

\section{Bollinger Bands and Pairs Trading} \label{S:BBPT}
Bollinger Band components are usually used with various other indicators in order to decide whether an asset is declining or trending.  The one 
quantitative strategy where the Bollinger Bands are often used solely on their own is that of pairs trading. In pairs trading \cite{DH2011}, where asset Z and asset 
X are the asset pair being traded,  
the quantity $y_{t} = ln(P_{z}/P_{x})_{t}$ is tracked over time as a time series.  It is assumed that $y_{t}$ is weakly stationary and therefore mean 
reverting.\footnote{The weak stationarity assumption is equivalent to assuming that the mean, $u_{t}$, of the process is constant. For this to be the case, mean 
reversion has to exist.} Therefore, when the quantity $y_{t}$ gets too high or too low, the expectation is that it will eventually return to its unconditional 
mean $u_{t}$.
Bollinger Bands are commonly used as a tool for exploiting this reversion behavior. Recall that $BBupper$ and $BBlower$ can be constructed from the Bollinger 
Band algorithm using the series ${y}_{t}$. Therefore, Bollinger Bands exploits reversion in the following manner: If at any time $t$, $y_{t}$ touches or crosses 
$BBupper$ ($BBlower$) at say $t^{*}$, then this is viewed a signal that, after $t^{*}$, $y_{t}$ is expected to decrease (increase) sometime in the relatively 
near future so a short (long) position is taken in the pair at $t^{*} + 1$.\footnote{A long position is defined as being long the asset that is in the numerator 
of the ratio of log prices and short the asset that is in the denominator. A short position is defined analogously.} Since $mave_{t}$ is the rolling mean 
estimate of the $y_{t}$  series, the crossing of $y_{t}$ back through $mave_{t^{**}}$ at some later time $t^{**}$ is used to indicate that the series 
$y_{t}$  has completely reverted and the position is closed out. 

\noindent For example, suppose that the series  $y_{t} = (ln(P_{z}/P_{x})_{t}$ is tracked over time and that it crosses  $BBlower_{t}$ at time $t=t^{*}$. Then, at time  $t^{*} + 1$, a long position 
is taken in asset Z and a short position, equal in  dollars to that amount taken in asset Z, is taken in asset X. This position is held until $y_{t}$ eventually crosses through the rolling mean 
$mave_{t^{**}}$ at some time $t^{**}$ in the future. The overall position entered into at $t^{*} + 1$ is then closed out by selling the long position in asset Z 
and buying back the short position in asset X. Conversely, suppose that over the course of time $y_{t}$ crosses through $BBupper$ rather than $BBLower$. Then, since we expect $y_{t}$ to 
decrease in the near future, we would go  short asset Z and long asset X. Then, when the $y_{t}$ process crosses back through the moving average, the overall position is closed out. 
For convenience going forward, we create an acronym for the Bollinger Band pairs trading strategy by referring to it as the BBPT strategy. 

\newpage

\subsection {The SAP 500 and the Nikkei: A Pairs Trading Example} \label{SS:pairtrade}
In order to illustrate the actual BBPT strategy using real data, we take the Standard and Poors 500 Index prices as asset Z and the Nikkei Index prices as asset X
and construct $y_{t}$, the difference between their log prices over the year 2004.\footnote{It is important to realize that in an actual trading scenario one would want to test that $y_{t}$ 
was cointegrated over some historical period immediately preceding the trading period 2004.} The Bollinger Bands generated by $y_{t}$ are based on a rolling window size of n = 20
and a bandwidth multiplier of $k = 2$. The pair trades generated during 2004 using these parameters are shown in Figure~\ref{FIG:2004_pair} in Appendix B. Each line segment in the figure, 
whether it red or green, represents the entry and exit of one trade which is initially generated by the touching or crossing of $BBUpper$ or $BBLower$ by $y_{t}$.. We will discuss the first, 
fifth and sixth trades in detail because these particular trades are representative of the typical behavior of the BBPT strategy.

\noindent Consider the first green line segment which represents the first trade. $BBUpper$ was crossed so the action 
taken was  to go short the paired asset by shorting the Standard and Poors Index and going long the Nikkei index. The time of the entry of short position is indicated by the arrow, 
(i.e. arrow always denotes the entry point) which is red to denote that the position was a short position. Clearly $y_{t}$ reverts to the 
moving average  very quickly and the position is then closed out. Since the moving average was crossed from above, this indicates that there was a profit from this short trade so the 
line segment is green. Finally, the diamond denotes the exit point. The diamond is red only because it is always the same color as its associated arrow which was red because a short position
was taken. 

\noindent 
The use of line segments, arrows and diamonds along with their respective colors allows for a large amount of information to be conveyed in
Figure~\ref{FIG:2004_pair}.  Also, because $y_{t}$ is defined as $y_{t} = (ln(P_{z}/P_{x})_{t}$, the return from any long pair trade\footnote{The return to a short trade is -1.0*
$(ln(P_{z}/P_{x})_{t^{**}} - (ln(P_{z}/P_{x})_{t^{*}}$.} entered into at $t^{*}$ and exited at $t^{**}$ trade is equal to  
$(ln(P_{z}/P_{x})_{t^{**}} - (ln(P_{z}/P_{x})_{t^{*}}$ which is conveniently equal to the vertical distance between the arrow and the diamond of the line segment.  
This relation is useful because one can then easily identify where there was a large positive return or large negative return. Green line segments with 
large vertical distances between their endpoints are signs of large  positive returns. Red line segments with large vertical distances between their endpoints 
represent large negative returns. 

\noindent Next, we consider the fifth line segment which represents the fifth trade. Just as was the case with the first trade, this is a short where the trade is 
short the SAP and long the Nikkei.  But unlike the first trade, the moving average is not crossed by $y_{t}$ until it is above the original entry point of  the trade. Therefore, the trade 
results in a loss and consequently the color of the line segment is red rather than green. The sixth trade in the figure is a long trade because BBlower was crossed but this trade also results 
in a loss because the $y_{t}$ series  did not cross through the moving average quickly enough. It crossed the moving average at a point lower vertically than where $y_{t}$ was at 
the entry point. $y_{t}$ was expected to increase after entry and cross the moving average from below rather than above but this did not happen. The vertical distance between the entry point 
and exit point represents the negative return of the trade. Notice that, for this trade, the endpoints of the line segment are green indicating that the trade was long the SAP and short the 
Nikkei.\footnote{Note that the weighted  sum of the the returns of all the trades in the figure is referred to as the return of the strategy over 2004 where the weights are proportional to 
the dollars allocated to each trade. For our purposes, we assume that all trades are given the same portfolio weight so that each trade weight =  1/number of trades.}

\noindent Since the horizontal axis of the plot  in Figure~\ref{FIG:2004_pair} represents time, the duration of any trade is simply the horizontal distance 
between when the trade opens (i.e. the arrow) and closes (i.e. the diamond).  One interesting aspect of the plot in Figure~\ref{FIG:2004_pair}
 that may not be obvious due to the scale of the axes is that the durations of the winning trades 
(i.e. green line segments) are consistently shorter than the durations of the losing trades (i.e. red line segments). In fact the average duration of the winning trades in  
is 8.6 and the average duration of the losing trades is 20.5. This duration behavior is not just specific to the use of the parameter values, n = 20 and k = 2. 
Consider Figure~\ref{FIG:all_pairs} in Appendix C.  Each of the eight plots represents the same pairs trading strategy simulated over different time periods using various 
combinations of the values of the window  size and the multiplier. The rolling window size parameter $n$ takes on the values of 20 and 30 while the width multiplier $k$ while
the width multiplier is either 1 or 2. The plots in Figure~\ref{FIG:all_pairs} clearly indicate that, in a BBPT strategy, the average duration of winning trades is 
a consistently shorter than the average duration of losing trades. Later on a more fundamental result will be proven concerning the return-duration behavior in the BBPT 
strategy.  This result is a key component  of the the Fixed Forecast Maximum Duration Bands pairs trading (FFDBPT) strategy which is 
discussed in the following section.

\newpage

\section{Fixed Forecast Maximum Duration Bands}\label{S:FFMDPT}
The use of Bollinger Bands in pairs trading goes back to the middle of the 1980's. The algorithm's popularity alone suggests that it has been at least reasonably successful 
in capturing mean reversion in paired assets with cointegrated price behavior.\footnote{Identifying a mean reverting pair of tradeable assets is a separate issue and will not be 
discussed here.} In this section, we create a series of successive links between various well known time series models which eventually lead back to Bollinger Bands. This successive linking will lead
to the development of a variant of Bollinger Bands called Fixed Forecast Maximum Duration Bands. But, in order to develop this variant, it is necessary to introduce the various time series models and  show how they are related. First we introduce the concept of exponential smoothing and its various properties. Next, we describe the 
Kalman filtering approach in some detail. Finally, we introduce a particularly simple Kalman filter called the random walk plus noise and make a connection 
between it and Bollinger Bands.

\subsection{Introduction to Simple Exponential Smoothing}
A well known forecasting method originally developed by Brown in the 1950's \cite{RGB59,RGB63} is that of simple exponential smoothing (SES). The method is
appropriate when it believed that the mean of the series might be changing over time but there is no trend or seasonality evident in the series. The method of SES 
smoothing takes the forecast for the previous period and adjusts it using the empirical forecast error. That is, the forecast for the next period is 
\begin{equation}\label{E:expsmth1}
\hat{y}_{t^{*} + 1} = \hat{y}_{t^{*}} + \lambda(y_t^{*} - \hat{y}_{t^{*}})
\end{equation}
The value of parameter $\lambda$ is restricted to be between 0 and 1 and is either determined empirically or known apriori based on the forecaster's previous experience. Of course, monitoring of
the parameter $\lambda$ is critical because the behavior of the series can change over time. We can re-write the forecast in the following manner in order to gain insight into what 
exponential smoothing is really doing:
\begin{equation}\label{E:expsmth2}
\hat{y}_{t^{*} + 1} =  \lambda{y_t^{*}} + (1-\lambda)\hat{y}_{t^{*}}
\end{equation}
By examining (\ref{E:expsmth2}), we can see that exponential smoothing is a model in which the forecast $\hat{y}_{t^{*}+1}$ is based on weighting the most recent observation $y_{t^{*}}$ with a weight
equal to $\lambda$ and the previous forecast with a weight equal to $(1-\lambda)$. Following Hyndman et al \cite{HKOS08}, the implications of exponential smoothing can be seen more easily 
if $\hat{y}_{t^{*}+1}$ is expanded by replacing $\hat{y}_t$ with its components as follows: 
\begin{eqnarray*}\label{E:expans1}
\hat{y}_{t^{*}+1} & = & \lambda y_{t^{*}}  + (1-\lambda)[\lambda y_{t^{*}-1} + (1-\lambda)\hat{y}_{t^{*}-1}] \\
& = & \lambda y_{t^{*}} + \lambda(1-\lambda)y_{t^{*}-1} + (1-\lambda^2)\hat{y}_{t^{*}-1}
\end{eqnarray*}
\noindent If this substitution is repeated by replacing $\hat{y}_{t^{*}-1}$ with its components, $\hat{y}_{t^{*}-2}$ with its components, and so on, the relation becomes:
\begin{eqnarray}\label {E:expans2}
\hat{y}_{t^{*}+1} & = & \lambda y_{t^{*}} + \lambda(1-\lambda)y_{t^{*}-1} + \lambda(1-\lambda)^{2}y_{t^{*}-2} + \lambda(1-\lambda)^{3}y_{t^{*}-3} \nonumber \\
& + & \lambda(1-\lambda)^{4}y_{t^{*}-4} + \cdots + \lambda(1-\lambda)^{t^{*}-1}y_{1} + (1-\lambda)^{t^{*}}\hat{y}_{1}
\end{eqnarray}
Therefore, $\hat{y}_{t^{*}+1}$ represents a weighted moving average of all past observations with the weights decreasing exponentially giving rise to the term ``exponential''
smoothing. Note that there is an initialization issue in that we need an initial value for $\hat{y}_{1}$. Usually, this value is taken to be the first observation or
some proportion of it and, if the series is long enough, the choice of this value should have a negligible effect on the predictions. For more elaborate methods for choosing the 
initial value, one should refer to \cite{HKOS08}.

\subsection{Simple Exponential Smoothing and the ARIMA(0,1,1)}\label{SS:sesarima}
The following discussion assumes that the reader has some familiarity with the ARIMA time-series modelling approach of Box and Jenkins. If this is not the case, 
then one is referred to \cite{BJR94} for a detailed description. First of all, it is well known that a forecasting equivalence exists between  particular exponential 
smoothing models and the mapped ARIMA model \cite{HKOS08,ML2000}. In fact, Muth  \cite{JM60} was the first of many authors to prove that SES is optimal for 
the ARIMA(0,1,1) process:
\begin{equation*}\label{E:arima011A}
(1-B)y_{t} = (1- \theta B)\epsilon_{t}
\end{equation*}
which can be re-written as
\begin{equation}\label{E:arima011B}
y_{t} =  y_{t-1} - \theta \epsilon_{t-1} + \epsilon_{t} 
\end{equation}  
\noindent Note that since the sign of $\theta$ is arbitrary, (\ref{E:arima011B}) can be re-written as 
\begin{equation}\label{E:arima011C}
y_{t} =  y_{t-1} + \theta \epsilon_{t-1} + \epsilon_{t} 
\end{equation}  
\noindent Also, in order for the ARIMA(0,1,1) model to be invertible, it is necessary to restrict $\theta$ so that $\theta \in (-1,1)$.
By SES being optimal, what is meant is that, if the parameter $\theta$ in ARIMA(0,1,1) process is known, then the SES method with parameter $\lambda = (1 - \theta)$
will give the same forecasts as the ARIMA(0,1,1) model. Unfortunately Muth's proof \cite{JM60} is not particularly transparent so we provide a simpler proof here.  
First, we write the one-step forecast for the 
ARIMA(0,1,1) model below:
\begin{equation}\label{EQN:arima}
y_{t^{*}+1} = y_{t} + \theta \epsilon_{t^{*}} + \epsilon_{t^{*}+1}
\end{equation}
Now, assuming that $\theta$ is known, if one wanted to calculate the one step ahead forecast using the ARIMA(0,1,1), the expectation of $\epsilon_{t+1}$ is zero  
so the forecasting equation becomes
\begin{equation*}
\hat{y}_{t^{*}+1} = y_{t^{*}} + \theta \epsilon_{t^{*}}
\end{equation*}
Therefore, generating the forecast, $\hat{y}_{t^{*}+1}$ requires estimating $\epsilon_{t^{*}}$ using $\hat{\epsilon}_{t^{*}}$. The estimate of $\hat{\epsilon}_{t^{*}} = (\hat{y}_{t^{*}} - y_{t^{*}})$ 
so the forecast becomes
\begin{eqnarray*}
\hat{y}_{t^{*}+1} & = & y_{t^{*}} + \theta (\hat{y}_{t^{*}} - y_{t^{*}}) \\
& = & (1-\theta)y_{t^{*}} + \theta \hat{y}_{t^{*}}
\end{eqnarray*}
But this is equation (\ref{E:expsmth2}) for simple exponential smoothing with parameter $\lambda = (1-\theta)$. Therefore, we have shown that the forecast of the ARIMA(0,1,1) model with 
parameter $\theta$ is identical to the forecast for SES with parameter $(1-\theta)$.
%%\newpage
\subsection{The Weighted Age In Simple Exponential Smoothing}
It should be emphasized that SES is not a time series model per se but rather a forecasting method
because there is no data generating process (DGP) underlying SES. Also, because of the invertibility  condition in the ARIMA(0,1,1), $\theta$ is  
restricted to lie between -1 and + 1 which  implies that $\lambda$ in SES is restricted to be between 0 and 2. 
In practical applications of SES, the $\lambda$ parameter is generally chosen to be between 0 and 1 in order to ensure that the weight given
to past observations decreases as the observations go further back in time. In fact, given $\lambda$, we can easily calculate the weighted average age of the observations used in the 
current forecast of SES. Notice that, in (\ref{E:expans2}), the weight given to an observation k periods ago, $y_{t-k}$, is $\lambda(1-\lambda)^{k}$.  Therefore, the 
weighted average age of the observations going into the current SES forecast at any time t is: \begin{eqnarray*}
\bar{k} & = & 0 \lambda + 1 \lambda (1-\lambda) + 2 \lambda (1-\lambda)^2 + \cdots \\
& = & \lambda \sum_{k=0}^{\infty} k (1-\lambda)^{k} \\
& = & \frac{(1-\lambda)}{\lambda}
\end{eqnarray*}

\noindent  A similar ``older data gets less weight' concept exists for the moving average forecast used in Bollinger Bands except that the decrease is more abrupt.  In the case of the moving 
average, the past observations that are of age n-1 periods or less are weighted equally with weight = 1/n. Any observations older than n-1 periods get a weight of zero. Therefore, for the 
moving average, we have:
\begin{eqnarray*}
\bar{k} & = & \frac{0 + 1 + 2 + \cdots + n - 1}{n} \\
& = & \frac{n-1}{2}
\end{eqnarray*}
\noindent An interesting question is whether there exists a parameter $\lambda$ in SES that will gives forecasts similar to that of the n period  moving average 
in Bollinger Bands. Brown \cite{RGB63} reasoned that, if the average age of the observations used in the SES forecast and the moving average forecast are the same,
then one would expect those models to give somewhat similar forecasts. Therefore, one can set the weighted age of the observations used in the current forecast 
of SES equal to the weighted age of the observations used in the moving average and solve for $\lambda$:
\begin{equation}\label{E:lambdatheta}
\frac{1-\lambda}{\lambda} = \frac{n-1}{2} \rightarrow \lambda = \frac{2}{n + 1}
\end{equation}

\noindent 
Figure~\ref{FIG:ewma_boll} in Appendix D shows the Bollinger Bands plotted along with the exponentially weighted moving average and its prediction intervals\footnote{The details pertaining to
the construction of the prediction intervals for exponential smoothing will not be discussed here. For details on the computation of the prediction intervals for the ARIMA(0,1,1) one is referred 
to \cite{CC00}.} when the weighted age relation, $\lambda = \frac{2}{n + 1}$, is used. The figure shows that the approximation is quite reasonable, particularly for the center line. This 
means that by using the relation $\lambda = 2/(n+1)$, the moving average associated with Bollinger Bands will provide a satisfactory approximation to the exponential smoothing model with 
parameter $\lambda$. 
%%===================================================
Below summarizes the connections made so far:
\begin{enumerate}
\item Exponential Smoothing and the ARIMA(0,1,1) are equivalent for $\lambda=1-\theta$
\item  The moving average is well approximated by exponential smoothing for $\lambda = 2/(n+1)$
\item 1 and 2 imply that the moving average is well approximated by an ARIMA(0,1,1) for $1-\theta = 2/(n+1)$
\end{enumerate}
\noindent This means that if we have an ARIMA(0,1,1) model with parameter $\theta$, then we can set $n = \frac{2}{1-\theta} -1$ in  the Bollinger Band moving average and this will provide a 
reasonable approximation to that ARIMA(0,1,1) model. These model connections are illustrated in Figure \ref{FIG:tslink} on page \pageref{FIG:tslink}.
\clearpage

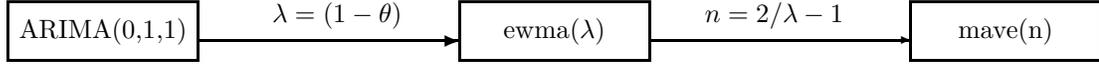
\begin{figure}[thb]
\setlength{\unitlength}{0.5cm}
\begin{picture}(1,5)
\thicklines
\put(2,1.5){\framebox(5.0,1.5){ARIMA(0,1,1)}}
\put(9.0,2.5){$\lambda = (1 - \theta$)}
\put(7,2.0){\vector(1,0){7.0}}
\put(14,1.5){\framebox(5.0,1.5){ewma($\lambda$)}}
\thinlines
\put(19,2.0){\vector(1,0){7.0}}
\put(20.5,2.5){$n = 2/\lambda - 1$}
\thicklines
\put(26,1.5){\framebox(5.0,1.5){mave(n)}}
\end{picture}
\vskip 0.1in
\caption{ 
The lines represent the connections between the models and
the transformations required to map one to the other. The thinner
line indicates that the relationship is approximate.}
\label{FIG:tslink}
\end{figure}

\noindent In order to make the final connection that leads to the Fixed Forecast Maximum Duration Bands pairs trading strategy, in what follows we briefly introduce state 
space models. We should point out that most of the introduction is taken from \cite{ML2000}. 

\subsection{Introduction To State Space Models}\label{SS:KF}
From the 1950's on, electrical engineers were particularly interested in the following problem in linear 
systems theory which is shown in Figure~\ref{FIG:filpic} on page \pageref{FIG:filpic}. Suppose we have an unobserved input signal at time t, $\boldsymbol{\theta}_{t-1}$ which is known as 
the system state. The state process evolves in accordance with a linear transform  of $\boldsymbol{\theta}_{t-1}$ to which is added a noise process $\boldsymbol{\omega}_{t}$. This process 
is described by the left hand side box in Figure~\ref{FIG:filpic}. The arrow from $\boldsymbol{\theta}_{t-1}$ to the box containing $\mathbf{F}_{t}^{\prime}$ represents the linear 
transformation of the system state producing the system output $z_{t} = \mathbf{F}^{\prime}_{t}\boldsymbol{\theta}_{t-1}$. Finally, added to $z_{t}$ is a noise process $\epsilon_{t}$ which 
results in the measurement process $y_{t}$. Only the sequence $\{y_{t}\}$ is observed. The linear system is described by equations (\ref{EQN:fileqn1}) and (\ref{EQN:fileqn2}). 
\begin{align}
{y}_{t} & = \mathbf{F}_{t}^{\prime}\boldsymbol{\theta}_{t-1} + \epsilon_{t}, \qquad
~~~~\epsilon_{t} \sim \mathrm{N}[\mathbf{0},V_{t}]
\label{EQN:fileqn1}
\\
\boldsymbol{\theta}_{t} & = \mathbf{G}_{t}\boldsymbol{\theta}_{t-1} +
\boldsymbol{\omega}_{t}, \qquad  ~ \boldsymbol{\omega}_{t}
\sim \mathrm{N}[\mathbf{0},\mathbf{W}_{t}]
\label{EQN:fileqn2}
\end{align}
\noindent with initial conditions
\begin{align*}
(\boldsymbol{\theta}_{0}~~|~~D_{0}) \sim \mathrm{N}[\boldsymbol{m}_{0},\boldsymbol{C}_{0}]
\end{align*}
%% FIGURE DESCRIBING HOW KALMAN FILTERING/STATE SPACE MODEL WORKS
%%=====================================================================================================================================

\noindent where $D_{0}$ denotes the information available at time zero. The normality assumptions on the error terms are not absolutely essential but they greatly simplify the 
inferential framework so they are usually imposed. The noise processes $\epsilon_{t}$ and $\boldsymbol{\omega}_{t}$ are assumed to be independent and the goal of the engineers was to 
produce an estimate of the unobserved system state $\boldsymbol{\theta}_{t}$ at time t, using the measurements, $y_{1},\ldots,y_{t}$. Then, when a new observation, $y_{t+1}$ is realized, 
a new estimate of the unobserved system state $\boldsymbol{\theta}_{t+1}$ should be obtained. This came to be known as the filtering  problem and was studied by electrical engineers for 
many years.

\newpage

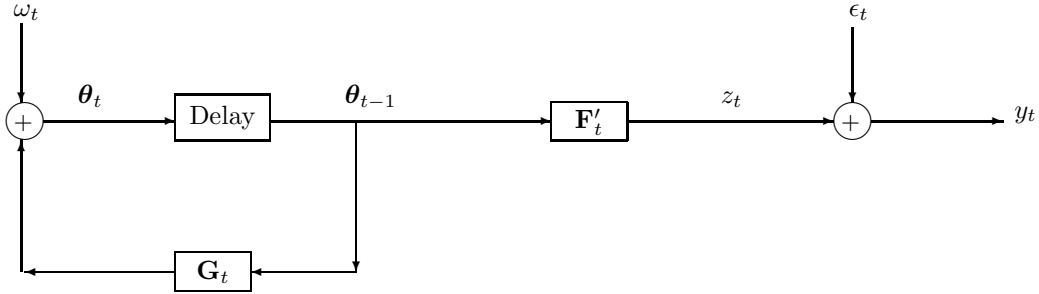
\begin{figure}[thb]
\setlength{\unitlength}{0.5cm}

\begin{picture}(15,10)
\put(2,2){\circle{1.0}}
\put(6,1.5){\framebox(2.5,1.2){Delay}}
\put(16,1.5){\framebox(2.0,1.0){$\mathbf{F}_{t}^{\prime}$}}
\put(24,2){\circle{1.0}}
%%\put(0,2){\vector(1,0){1.5}}
\put(2.5,2){\vector(1,0){3.5}}
\put(8.5,2.0){\vector(1,0){7.5}}
\put(18.0,2.0){\vector(1,0){5.5}}
\put(24.5,2.0){\vector(1,0){3.5}}
\put(23.7,1.80){\small {+}}
\put(3.4,2.5){$\boldsymbol{\theta}_{t}$}
\put(10.5,2.5){$\boldsymbol{\theta}_{t-1}$}
\put(1.7,1.80){\small {+}}
\put(20.5,2.5){$z_{t}$}
\put(28.3,2.1){$y_{t}$}
\put(10.8,2.0){\vector(0,-1){4.0}}
\put(10.8,-2.0){\vector(-1,0){2.8}}
\put(6,-2.5){\framebox(2.0,1.0){$\mathbf{G}_{t}$}}
\put(1.9,-2.0){\vector(0,1){3.5}}
\put(6.0,-2.0){\vector(-1,0){4.0}}
\put(24,4.5){\vector(0,-1){2.0}}
\put(23.9,4.8){${\epsilon}_{t}$}
\put(1.7,4.8){${\omega}_{t}$}
\put(1.9,4.6){\vector(0,-1){2.1}}
\end{picture}

\vskip 0.5in

\caption{ 
A schematic diagram representing the
filtering problem in electrical engineering.
}\label{FIG:filpic}

\end{figure}

%% END OF FIGURE DESCRIBING HOW KALMAN FILTERING/STATE SPACE MODEL WORKS
%%=====================================================================================================================================

\noindent In an extremely important contribution to the engineering literature,  Kalman \cite{KF60} developed a recursive scheme for updating $\boldsymbol{\theta}_{t}$ optimally
each time a new observation $y_{t}$ is realized. These recursions became known as the Kalman filter recursions\footnote{The recursions are somewhat complicated so they are not given here but they can be found in ~\cite{AJ70}, ~\cite{AH92} and ~\cite{ML2000}.} and the system described by (\ref{EQN:fileqn1}) and (\ref{EQN:fileqn2}) became known as the state space 
formulation. 

\noindent Then, in the early 1970's, Harrison and Stevens \cite{HS76} bridged the gap between the statistical community and the engineering community by showing that the state space
formulation could be used by statisticians to build and estimate models that were already very popular in the statistical  literature. For example, they showed that
if one took the state space model and let the matrix $\mathbf{G}_{t}$ be the identity matrix, then the model was equivalent to a time varying coefficient regression model.
This brought state space models into the statistical community and led to various specific state space models one of which is described in Section \ref{SS:rwpn}.

%%===========================================================================================================================================================
%%====================================================================================================================================================
%%====================================================================================================================================================
%%====================================================================================================================================================

%%%%%% MORE WRITING HERE ABOUT SIGNAL TO NOISE RATIO RELATION AND GRAPHICS 

\newpage
\subsection{A Simple State Space Formulation: The Random Walk Plus Noise Model }\label{SS:rwpn}
\noindent Suppose that we have the state space formulation represented by (\ref{EQN:fileqn1}) and (\ref{EQN:fileqn2}) where $\mathbf{F}_{t} = 1$, $\mathbf{G}_{t} = 1$, and 
$\boldsymbol{\theta_{t}}$ is is a scalar equal to $\mu_{t}$.  Then, the state space formulation reduces to what is termed the random walk plus noise state space model (RWPN) shown below: 
\begin{align}
y_{t} & = \mu_{t-1} + \epsilon_{t} \label{EQN:HWrwpn1} \\
\mu_{t} &  = \mu_{t-1} + \eta_{t}  \label{EQN:HWrwpn2}
\end{align}
\noindent where
\begin{align*}
\epsilon_{t} & \sim \mathrm{N}(0,\sigma^2_{\epsilon}) \\
\eta_{t} & \sim \mathrm{N}(0,\sigma^2_{\eta})
\end{align*}
\noindent As noted previously, $\epsilon_{t}$ and $\eta_{t}$ are assumed to be independent.
\noindent Eliminating the system variable $\mu_{t}$ and creating a stationary model by differencing $y_{t}$ gives:
\begin{equation}\label{EQN:tri}
{\triangle}y_{t} = \eta_{t-1} + (\epsilon_{t} - \epsilon_{t-1})
\end{equation}

\noindent Now, we can easily calculate the model autocorrelations , $\gamma_k$, for each $k$:
\begin{align*}
\gamma_{1} & = \frac{-\sigma^2_{\epsilon}}
{\sigma^2_{\eta} + 2\sigma^2_{\epsilon}} \\
\gamma_{2} & = \gamma_{3} = \cdots = 0
\end{align*}
\noindent Since $\mathrm{E}(\triangle{y_{t}})  = 0 $, and the autocorrelations are zero after lag 1, the RWPN model is statistically equivalent to an ARIMA(0,1,1) model.
\noindent Note that since all the variances are required to be greater than zero, we can see by inspection that for this random walk plus noise model, $ -0.5 < \gamma_1 < 0$.
This implies that the parameter space for $\theta$ in the equivalent ARIMA(0,1,1) is restricted. In fact, if we equate $\gamma_{1}$ in the random walk plus noise model with 
$\gamma_{1}$ in the ARIMA(0,1,1) model, then the parameter ${\theta}$ is forced into the range $-1 < {\theta} < 0$.\footnote{There is another state space model called the 
single source of error state space model which does not impose this restriction. See \cite{HKOS08,ML2000} for details.} 

\noindent  We need to derive the exact relation that maps the random walk plus noise model to the ARIMA(0,1,1). First of all, it is easy to show
that the autocorrelation at lag one of the ARIMA(0,1,1) defined in equation (\ref{EQN:arima}) is $\theta/(1+\theta^{2})$. If we define
the signal to noise ratio in the random walk plus noise model as $ q = \sigma_{\epsilon}^2/\sigma^2_{\eta}$ and equate the lag one autcorrelations of each model,
we obtain the following mapping between the two models:
\begin{equation}\label{EQN:rwpnmap}
\theta = (\sqrt{q^2 + 4q} - 2 - q)/2
\end{equation}
\noindent Therefore, given the signal to noise ratio $q$ , we can find the equivalent ARIMA(0,1,1) using the mapping in equation (\ref{EQN:rwpnmap}).
This connection between the RWPN and the equivalent ARIMA(0,1,1) allows us to modify Figure \ref{FIG:tslink} on page \pageref{FIG:tslink} by including the random 
walk plus noise model. The resulting Figure is shown below and illustrates that a rolling window moving average using $n$ as the window size can be viewed as an 
approximation to the random walk noise model with a signal to noise ratio equal to $q$. What we have shown is that the rolling moving average approach used in Bollinger 
Bands, initially thought to be a ``poor man's time varying regression coefficient model'', may not  be as poor as originally thought. More importantly, in the next section 
we will see how this connection between the RWPN and the moving average leads to an interesting modification of Bollinger Bands. 
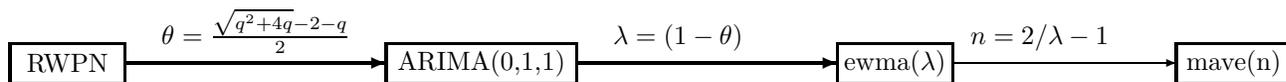
\begin{figure}[thb]
\setlength{\unitlength}{0.5cm}
\begin{picture}(0.5,5)
\thicklines

\put(-1,1.0){\framebox(3.0,1.0){RWPN}}
\put(3,2.0){$\theta = \frac{\sqrt{q^2 + 4q} - 2 - q}{2}$}
\put(2,1.5){\vector(1,0){7.0}}

\put(9,1.0){\framebox(5.0,1.0){ARIMA(0,1,1)}}
\put(15,2.0){$\lambda = (1 - \theta$)}
\put(14,1.5){\vector(1,0){7.0}}

\put(21,1.0){\framebox(3.0,1.0){ewma($\lambda$)}}
\thinlines
\put(24,1.5){\vector(1,0){6.0}}
\put(24.5,2.0){$n = 2/\lambda - 1$}

\thicklines
\put(30,1.0){\framebox(3.0,1.0){mave(n)}}
\end{picture}
\vskip 0.1in
\caption{ 
Given the various mappings, the moving 
average construction can be viewed as a approximation to the 
random walk plus noise model.}
\label{FIG:tslinkb}
\end{figure}

\newpage

\subsection{FFMDPT: A Bollinger Band Variant}
In the previous section,  we showed that the moving average in Bollinger Bands can viewed as an  approximation to the random walk plus noise state space model 
with a signal to noise ratio equal to $q$. First, in order to make the state space notation consistent with the Bollinger Band notation previously used 
where $\beta_{t}$  has been viewed as being synonymous with $mave_{t}$, we modify equations (\ref{EQN:HWrwpn1}) and (\ref{EQN:HWrwpn2}) representing the RWPN  
by replacing $\mu_{t}$ with $\beta_{t}$ everywhere. This results in the RWPN model equations below:
\begin{align} 
y_{t} & = \beta_{t-1} + \epsilon_{t} \label{EQN:HWrwpn1b} \\
\beta_{t} &  = \beta_{t-1} + \eta_{t}  \label{EQN:HWrwpn2b}
\end{align}
\noindent where
\begin{align*}
\epsilon_{t} & \sim \mathrm{N}(0,\sigma^2_{\epsilon}) \\
\eta_{t} & \sim \mathrm{N}(0,\sigma^2_{\eta})
\end{align*}

\noindent Equations (\ref{EQN:HWrwpn1b}) and (\ref{EQN:HWrwpn2b}) shed light into what the BBPT strategy is doing from a state space viewpoint. Again, we consider 
Figure~\ref{FIG:2004_pair} on page  \pageref{FIG:2004_pair}.  From the figure, we can see that after a trade entry in the BBPT strategy, the moving  average, $\hat{\beta_{t}}$, is the dynamic  forecast for $y_{t}$ and in this sense, its forecasted equilibrium price. Therefore, if we view the moving average in 
Bollinger Bands as an approximation to the random walk plus noise model in (\ref{EQN:HWrwpn1b}) and  (\ref{EQN:HWrwpn2b}), then, after a trade entry, 
the Bollinger Band algorithm  continues to receive the new data, $y_{t}$, and the estimate of $\beta_{t}$ in (\ref{EQN:HWrwpn2b}) is updated 
accordingly. Note that a statistically consistent update of $\beta_{t}$  requires that the estimated signal to  noise ratio, $q$, is remaining constant.

\noindent Now, rather than continuing to update the estimate $\hat\beta_{t}$ as if the signal to noise ratio was still $q$ after trade entry, we can make an 
alternative assumption. Suppose that, immediately after trade entry, we view the future $y_{t}$ as missing and then continually forecast as if there was no 
longer any future 
$y_{t}$ available. This assumption could be quite reasonable because any trade entry implies that some kind of unusual observation or outlier with respect
to the current state of the system has been observed. 
Once an outlier is observed during the evolution of a state space model, there is little reason to assume that the state space model is in the same state with respect 
to $q$ as it was before trade entry. Once the  assumption is made that future $y_{t}$ are missing after trade entry, the observation equation (\ref{EQN:HWrwpn1b}) 
no longer exists and the original random walk plus noise model represented by (\ref{EQN:HWrwpn1b}) and (\ref{EQN:HWrwpn2b}) reduces to a pure random walk model for 
$\beta_{t}$ :
\begin{equation} \label{EQN:RW}
\beta_{t} =  \beta_{t-1} + \eta_{t} 
\end{equation}

\noindent  Note that if $\beta_{t}$ is evolving as a random walk, then this implies that the optimal forecast to make after trade entry is $\hat{\beta}_{tradeentry}$ itself, the estimated value of $\beta_{t}$ at 
the time of trade entry. This constant forecast is the first modification we make to the Bollinger Band algorithm. Rather than using the moving average, $mave_{t}$ as the forecast at each time 
$t$ after trade entry, we use the  $mave_{tradeentry}$, namely the known moving average estimate at time t,  as the future forecast at all times t. 
In this way, the forecast for where the  process will revert is a horizontal line segment (i.e. a constant ) starting at $mave_{tradeentry}$ and we call 
this forecast the ``fixed forecast''.  

\noindent An advantage to using the ``fixed forecast'' variant of BBPT (e.g. FFPT) is that if $y_{t}$ does cross through the fixed forecast, the return 
generated by the trade  will always be greater than the return generated by the identical trade in the BBPT strategy. This is because, in a long FFPT trade,
the fixed forecast at time t will always greater than the $mave_{t}$ (i.e: the forecast in the BBPT strategy ) and, in a short FFPT trade, the fixed forecast will 
always less than the $mave_{t}$. Unfortunately, the FFPT algorithm also introduces a serious problem. Clearly, since we are assuming that the $\beta_{t}$ process is
a random walk, the variance of the fixed forecast k periods out is $\sigma^2_{\eta}*k$ so as one goes further and further out, the forecast variance increases 
linearly with k. More problematic is the fact that there is a non-zero probability that the $y_{t}$ process may never revert and cross through the horizontal 
forecast. Conversely, in the case of Bollinger Bands, the possibility of the $y_{t}$ not crossing the forecast is extremely unlikely because the moving average 
forecast is a function of the $y_{t}$ process and essentially tracks it. In fact, in the BBPT strategy, the only scenario in which the $y_{t}$ 
process will not revert to the moving average is one in which the $y_{t}$ process trends permanently in either direction. Clearly this permanent trending 
scenario is extremely unlikely in practice. In the next section, we develop a FFPT trade exit mechanism which remedies its ``trade may never exit'' problem. 

\newpage

\subsection{Restricting the Trade Duration in The Fixed Forecast Variant}
\label{SS:resdur}
Recall that, in Section~\ref{SS:pairtrade}, we saw that the average duration of losing trades in BBPT strategies was greater than the average duration of 
winning trades. Although this relationship was only shown empirically, we formalize the duration-return relationship in the following theorem:

\begin{theorem}  \label{THM:theorem1b}
Assume that the rolling window size in the BBPT strategy  = n,  the band width multiplier = k and that $y_{t}$ touches $BBlower$ at $t = t^{*} - 1$ 
so that a long trade is generated at $t = t^{*}$. Then the total return of this trade is non-negative if and only if the duration of the trade is less 
than or equal to n; i.e. the trade is exited at a time less than or equal to $t^{**} = t^{*} + n - 1$. This result is independent of the bandwidth multiplier 
parameter k.
\end{theorem}

\begin{proof}
See Appendix E.
\end{proof}
\noindent We should point out that Theorem \ref{THM:theorem1b} assumes that slippage does not occur during entry nor exit. By absence of
slippage, we assume that during the time between the trade entry signal and the trade entry,  price erosion does not occur. Similarly,
during the time between the trade exit signal and trade exit, we also assume that price erosion does not occur. If either of these 
assumptions are not true, then Theorem \ref{THM:theorem1b} will hold only approximately. In fact, if we return to the empirical evidence in 
Section~\ref{SS:pairtrade}, we see that the average losing trade durations in 2007 for k = 1 and k = 2 and n = 30 were actually slightly less than n = 30. 
This is inconsistent with Theorem \ref{THM:theorem1b} but the inconsistency only occurs because the  BBPT simulations assume that the there is a one day lag 
between the entry signal and the entry and a one day lag between the exit signal and the exit. For the k=1 and k=2 cases in 2007 using n = 30, 
slippage did occur and caused the average duration of the losing trades to be less than the rolling window size n = 30.

\noindent Theorem \ref{THM:theorem1b} suggests that a reasonable exit time for a trade in FFPT 
is the rolling window size itself.\footnote{Using technical analysis terminology, an exit rule based on a pre-determined length of time is referred to as 
a time 
stop.} The logic behind this idea is that once the trade duration becomes greater than the window size, by Theorem \ref{THM:theorem1b}, the trade 
cannot possibly have an overall positive return. Therefore, intuitively the rolling window size serves as  a reasonable time to exit the trade. Also, exiting at a 
time t equal to the rolling window size remedies the original ``trade may never exit'' problem associated with the FFPT approach.  Therefore, restricting the 
maximum duration of all FFMDPT trades to the window size $n$ is the second and final modification to the BBPT strategy and results in a variant we call the Fixed 
Forecast Maximum Duration Pairs Trading strategy (FFMDPT). To summarize, the first component of FFMDPT is that the forecast is a constant equal to the original 
moving average at trade entry.  Secondly, assuming that the window  size = n, then, if a trade has not exited by n periods, the trade is exited at the nth period. 
Notice that the parameter $n$ in FFMDPT has a different role from its role in BBPT. In the FFMDPT strategy, the parameter $n$ is used to calculate the candidate 
entries $BBUpper$ and $BBLower$ and the exit point $mave_{t}$. In FFMDPT, the entry signals are the same as those in BBPT but the exit signal differs because the 
maximum duration of any trade is equal to the rolling window size. Figure~\ref{FIG:three_plots} in Appendix F illustrates the simulation results of the BPPT 
strategy and the FFMDPT strategy for the 2004 SAP-Nikkei data  using $n=20$ and $k=2$. Note that in this specific example, the return generated by the FFMDPT 
strategy is about three percent greater than that of the BBPT strategy but this will not always be the case in general. We illustrate specific differences between 
FFMDPT and BBPT with two detailed examples in the section that follows.

\subsection{BBPT versus FFMDPT: Two Examples}\label{SS:compare}
\noindent In what follows, we construct two simulations to illustrate how the different exit rules can effect the relative performance of the BBPT and FFMDPT 
strategies. The first simulation ran through the full year of 2004 using $n=20$ and $k=2$. In this simulation, the FFMDPT strategy outperforms the BBPT strategy. The 
second simulation ran through the full the year of 2005 again using the parameters $n=20$ and $k=2$. In this simulation, the FFMDPT strategy underperforms the BBPT 
strategy. The plots displaying the results of the first simulation and second simulation are shown in Appendix H on page \pageref{FIG:two_plotsa} and
\pageref{FIG:two_plotsb} respectively. In what follows, we analyze specific trades associated with the two simulations in detail.

\noindent Consider the 2004 simulation. The FFMDPT strategy outperforms the BBPT strategy by approximately three percent. Although it may not be 
obvious from the plot, although the three trades in February, April and June generates similar returns in both strategies, the trades in the FFMDPT strategy generate 
slightly larger returns than the corresponding trades in the BBPT strategy. This is due to the fixed forecast creating larger vertical distance 
between the entry point and the exit point. Therefore, conditional on $y_{t}$ touching or crossing the fixed forecast in the FFMDPT strategy, the associated 
return will be larger than the return of the same trade in BBPT. Next, consider the fifth trade in both strategies at the beginning of August. In the FFMDPT 
strategy, the trade has a longer duration because it needs to reach the horizontal forecast before it exits. Consequently, the FFMDPT August trade 
results in a small loss.  On the other hand, the same August trade in the BBPT strategy exits quickly because the moving average forecast is reached 
sooner than the fixed horizontal forecast. Therefore, the BBPT trade exits earlier than the FFMDPT resulting in a larger loss. In this particular simulation, 
because FFMDPT trades are held until the fixed forecast is reached, this resulted in larger winning trades and smaller losing trades compared to the same trades 
in BBPT resulting in relative overperformance of the FFMDPT strategy in 2004.

\noindent Next consider the 2005 simulation. The FFMDPT underperforms the BBPT strategy by approximately 1.2 percent. The first trade in the BBPT strategy 
triggered at the beginning of February generates a positive return but the same trade in FFMDPT generates a negative return because the horizontal forecast is 
never reached so the trade is exited when it reaches the maximum duration. Unfortunately the maximum duration is reached after the $y_{t}$ process has hit the 
moving average and then decreased. The FFMDPT trade exits at a significant loss of almost 1\% while the same trade in the BBPT strategy generates a positive 
return. The rest of the trades in 2005 generate approximately the same return in both strategies so the underperformance of the FFMDPT strategy is 
essential due to the behavior of the FFMDPT trade at the beginning of February.

\subsection{FFMDPT versus BBPT: An Optimized Simulation}
\noindent  Although the 2004 and 2005 simulations highlight the differences between BBPT and FFDMPT, unfortunately the results are dependent on 
whether the asset pair being traded, namely the SAP-Nikkei, exhibits mean reversion.\footnote{The mean reversion assumption can be tested statistically by 
using the Two Step Engle-Granger test for cointegration \cite{DH2011}.}  There exist various methods for checking or testing  its existence historically, 
but whether the mean reversion behavior persists in the future is also uncertain. The point is that any comparison of the effectiveness of pairs trading strategies 
is confounded with the possibly unstable mean reverting behavior of the asset pair being traded. The 2004 and 2005 simulations of the SAP-Nikkei Index  pair 
assume that the traded asset pair, SAP-Nikkei, is cointegrated and this assumption is critical to the performance comparison of the two strategies, BBPT and FFMDPT. 
In fact, the FFMDPT strategy is even more dependent on mean reversion behavior because it requires that the $y_{t}$ process revert all the way back to the forecast 
at trade entry. In the case of the BBPT strategy, a positive return can be generated by a trade even when the $y_{t}$ process does not return to the trade entry forecast.
Nevertheless, as long as we are aware that the mean reversion issue makes the results less definitive,  we can compare the  performance of the respective pairs 
trading strategies assuming that the SAP-Nikkei  Index pair is mean reverting.  

\newpage

\noindent As explained earlier, the functionality of the parameter $n$ in the FFMDPT strategy is different from that of the parameter $n$ in the BBPT strategy. 
Therefore using the same parameter combination of $n$ and $k$ when comparing the two strategies is not necessarily correct. Given the difference in 
functionality, it is more reasonable to view the parameter $n$ as being a different parameter in each strategy, namely $n_{FFMDPT}$ and $n_{BBPT}$ respectively.  
A more robust methodology for analyzing strategy performance is to first find the optimal parameters $n_{BBPT}$ and $n_{FFMDPT}$ for the respective 
pair strategies over some historical period called the in sample period. These optimized values can then  be used out of sample and the out of sample performance 
compared. Using R \cite{R11}, this methodology was implemented for the  SAP-Nikkei index pair for each year from 2003-2011. Three simulations were run for which the 
value of the parameter $k$ was 1,1.5 and 2.0.  For each year, the optimal values of $n_{BBPT}$ and $n_{FFMDPT}$ were found by searching over values from 10 to 50 in 
steps of 1. These optimal values were then used in the following year and the performance in that year was calculated. The results are shown in Tables I, II and III 
in Appendix I where 2003-4 denotes that the in sample period was 2003 and the out of sample period was 2004.\footnote{For the 2010-11 period only the first four 
months of data in 2011 was available.} Table I indicates that there is no clear performance difference for $k=1$. In four of the eight years, 
BBPT outperformed FFMDPT with the only large difference occurring in 2008 when FFMDPT outperformed BBPT by almost ten percent. In Table II when $k=1.5$ was used, 
the BBPT strategy significantly outperforms the FFMDPT strategy. In fact, in six out of the eight years, the BBPT returns are larger and, in 2004, 2005 and 2007, 
more than five percent larger. Finally, for k = 2, the BBPT strategy again outperforms the FFMDPT strategy in six out of the 8 years. In this case, the 
differences generally hover around the 2\% level. 

\noindent In addition to the question of the asset pair possessing mean reverting behavior, another problem with assessing relative performance using an optimization
approach is the possibility of overfitting.  A current pitfall of the optimization approach is that the model may stick too closely to the data over which the 
optimization was performed. 
Consequently, the model ends up learning irrelevant details of the in sample data which leads to poor generalization when the parameters are used on the out of 
sample data. A informal way of thinking about overfitting is that  because such a fine search was used to find the optimal value of the respective 
parameters in the in sample period, it might be the case that a diamond was found in sample but corrodes quickly when used out of sample. There are various 
methodologies in the literature that attempt to deal with the issue of overfitting but these will not be discussed in this study. For an interesting discussion 
of overfitting and methods for lessening its effect, the reader is referred to \cite{MO11}.  One possible way to reduce the amount of overfitting is to optimize 
more frequently. For example, rather than optimizing the parameters over one year and then calculating the performance in the following year, we could use a 
shorter in sample and out of sample period such as three months.\footnote{Note that three months is the shortest sample period that could be used because the 
optimization uses a grid from n = 10 to n = 50. Since we want the n=10 simulation and the n=50 simulation to start at the same time, a minimum of fifty data 
points are required to calculate the first moving average.} Another possible way to decrease the possibility of overfitting is to reduce the parameter grid 
size search. For example, one could decrease the number of candidate values of $n_{FFDMPT}$ and $n_{BBPT}$ by only allowing values from 1 to 50 in steps of say 5.

\noindent In summary, although the simulation results provide some evidence that the BBPT strategy outperforms 
the FFMDPT strategy, there are issues that limit the strength of this evidence. The first is the question of the existence of mean reversion behavior in the 
specific asset pair analyzed. The second is the possibility of overfitting when we are optimizing over the parameters $n_{FFMDPT}$ and $n_{BBPT}$.

\section{Conclusions and Future Research}\label{S:FUTRES}
This article contributed towards reconciling the relationship between Bollinger Bands and time series models. First we showed that,
aside from requiring a slight modification to the prediction bands, the Bollinger Band components can be mapped exactly to the outputs of a classical regression 
model. This mapping provides a statistical foundation for Bollinger Bands and eliminates the algorithmic and ad hoc reputation it has had until now.
Also, through the use of a series of relations linking various time series models, we were able to show that Bollinger Bands can be viewed as a reasonable
approximation to the random walk plus noise (RWPN) state space model. 

\noindent Next we proved an interesting result connecting the return-duration relationship in Bollinger Bands. 
Although the result of theorem was proven with respect to Bollinger Bands, its importance lies in the fact that it holds for all cases where one uses distance from 
a moving average as an entry signal and reversion to the moving average as an exit signal. Then by modifying the underlying assumption of the approximate RWPN model 
and using the return-duration result for moving averages, we developed a variant  of Bollinger Bands called Fixed Forecast Maximum Duration Bands (FFMDPT). In the case 
of the SAP-Nikkei data from 2003-2010, FFMDPT generates returns that generally underperform the BBPT strategy particularly when $k=1.5$. At the same time, there are l
imitations to the strength of the result when doing such a strategy comparison and some possible remedies for the limitations were discussed.

\noindent One future research possibility is to compare the FFMDPT and BBPT strategies but optimize over the parameters $n$ and $k$ jointly. Although it was 
shown that the return-duration proof was independent of $k$, the entry signal in FFMDPT and BBPT is still dependent on $k$. Therefore, optimizing the two parameters 
$n$ and $k$ jointly may lead to different results. Also, it would be useful if additional pairs were investigated so that the performance results were not specific 
to the SAP-Nikkei Index pair. Tests for cointegration of the pairs could be done in order to ensure that only pairs that were cointegrated historically were included 
in the performance comparison.

\noindent Another research area would be to take the standard Bollinger Band pairs trading strategy and use the theorem result to change the standard exit
rule by exiting the trade  whenever the trading duration is equal to the rolling window size.  Although this type of exit rule does mean that one has given up the chance
to gain back some of the loss from the losing trade, could possibly utilize capital more efficiently by restricting live trades 
to only those which have a chance of being profitable. Clearly, this type of strategy implies that the Bollinger Band parameters, $n$ and $k$ may need to be changed also. 

\noindent Finally, another research direction would involve implementing a statistically consistent approach for capturing reversion behavior.  Bollinger Bands 
are still only an approximation to a random walk plus noise model because the algorithm is such that observations with an age greater than the rolling window size 
ago are given a weight of zero during estimation. Rather than using Bollinger Bands to implement the pairs trading strategy, the Kalman filter approach could be used 
directly by utilizing the recursive updating equations. Implementing the Kalman filter approach would  eliminate the  parameter $n$ but would necessitate estimating
the observation variance and the system variance (i.e. $q$). The possibility of overfitting would no longer be an issue because there would no longer be a need for 
optimization over the rolling window size parameter $n$. At the same time, one would still need to do an investigation into what the optimal re-estimation frequency for 
the RWPN variances should be.

\newpage
%%===========================================================
\section{Acknowledgements}
The author wishes to acknowledge Professor Chanseok Park, Clemson University, Dept of Mathematical Sciences and Louis Kates, GKX Associates 
for many helpful discussions that improved the quality of this article. Thanks also goes to John Bollinger of Bollinger Capital Management for 
providing the SPX-Nikkei data and also the impetus for this study when we  met at the R-Finance conference in May, 2011. Finally, an 
acknowledgment goes to my advisor, Professor Keith Ord, Georgetown University, McDonough School of Business for spurring my initial interest 
in time series models.
%%===========================================================

%%\newpage

%%================================================================================================================================

\clearpage

\noindent \textbf{Mark Leeds} is a statistical consultant in New York City specializing in financial econometrics and time series analysis. He 
received a B.S in Operations Research from Columbia University in 1988; the M.S. in Operations Research and Statistics from Rensselaer 
Polytechnic University in 1990; and a Ph.D. in statistics from Pennsylvania State University in 2000. His work focuses on using econometric 
methodologies to uncover anomalies in the capital markets. His email address is markleeds2@gmail.com.

%%==================================================================
\clearpage
\part*{\LARGE\centering Appendices}
\appendix  

%% BOLLINGER BAND PAIRS TRADING CONSTRUCTION
%%==================================================================
\section{~~~~~Illustrating the Bollinger Band Construction}
\begin{figure}[htb] 
\begin{center}
  \includegraphics[scale=0.70]{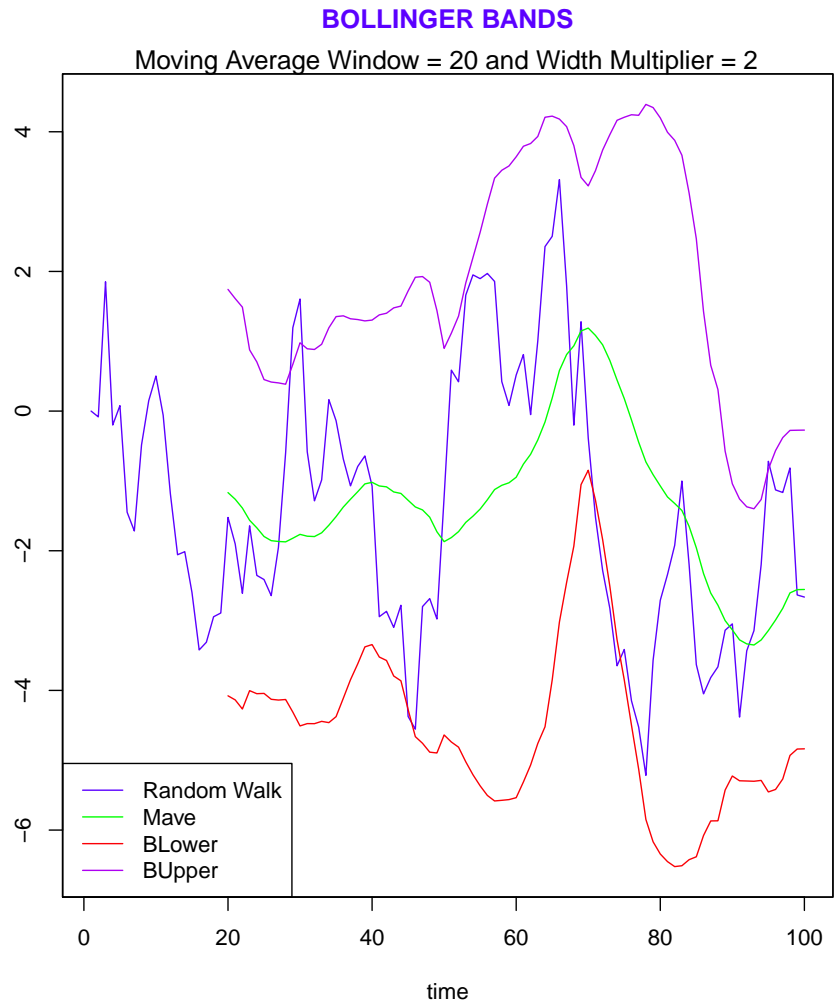}
\end{center}
\caption{ 
A simple example that illustrates the construction of Bollinger Bands. A random walk series
was generated initially. The green center line is the n = 20 day moving average of the random walk series and 
BBupper and BBlower are k = 2 standard deviations above and below the center line. }
\label{FIG:boll_pic1}
\end{figure}

%%========================================================================================================================================
%% BOLLINGER PAIRS IN PAIRS TRADING EXAMPLE
\clearpage
\section{~~~Illustrating the Use of Bollinger Bands in Pairs Trading}
\begin{figure}[htb] 
\begin{center}
  \includegraphics[scale=0.9]{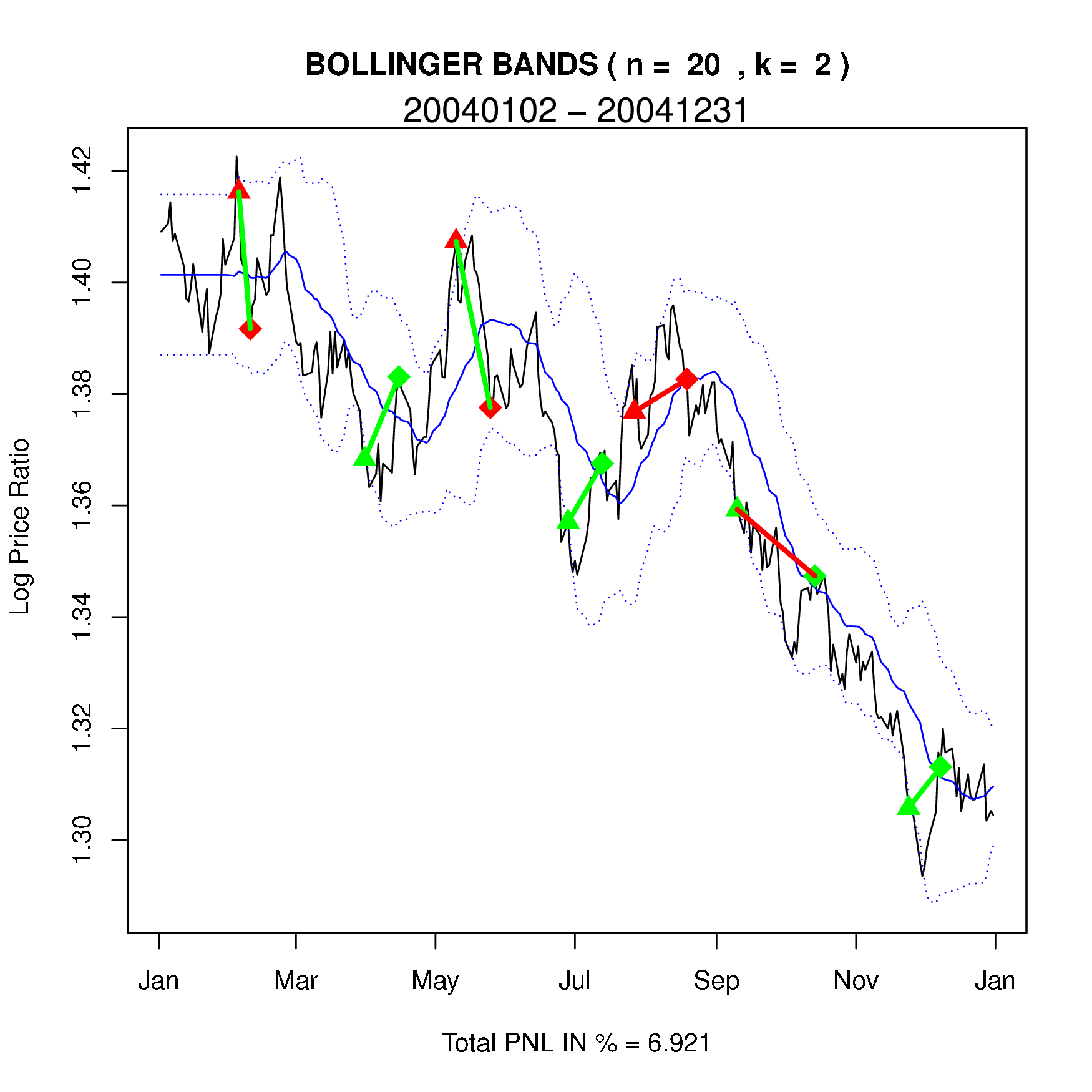}
\end{center}
\caption{ 
The use of Bollinger Bands in a pairs trading strategy. The
center line in the figure is the n = 20 day moving average and BBupper and BBlower are k = 2 standard
deviations above and below the center line. The crossing of BBupper from below (i.e. red arrow) or BBlower from above (i.e. green arrow) triggers a 
short trade (i.e. red arrow) or long trade (i.e. green arrow ). The position is held (i.e. line extends) until the series reverts to the center line. 
resulting in either a winning trade (i.e. green line)  or losing trade (i.e. red line).  Note that Z = SAP Index and X =  Nikkie Index so that the plotted series 
is  $y = ln(P_{z}/P_{x})_{t}$.
}\label{FIG:2004_pair}
\end{figure}

%% EIGHT PLOTS SHOWING HOW DURATIONS OF LOSERS ARE LONGER THAN DURATIONS OF WINNERS
%%===================================================================================
\clearpage
\section{BBPT Winning Trades Have Shorter Durations Than Losing Trades}
\begin{figure}[h!]
\centering\includegraphics{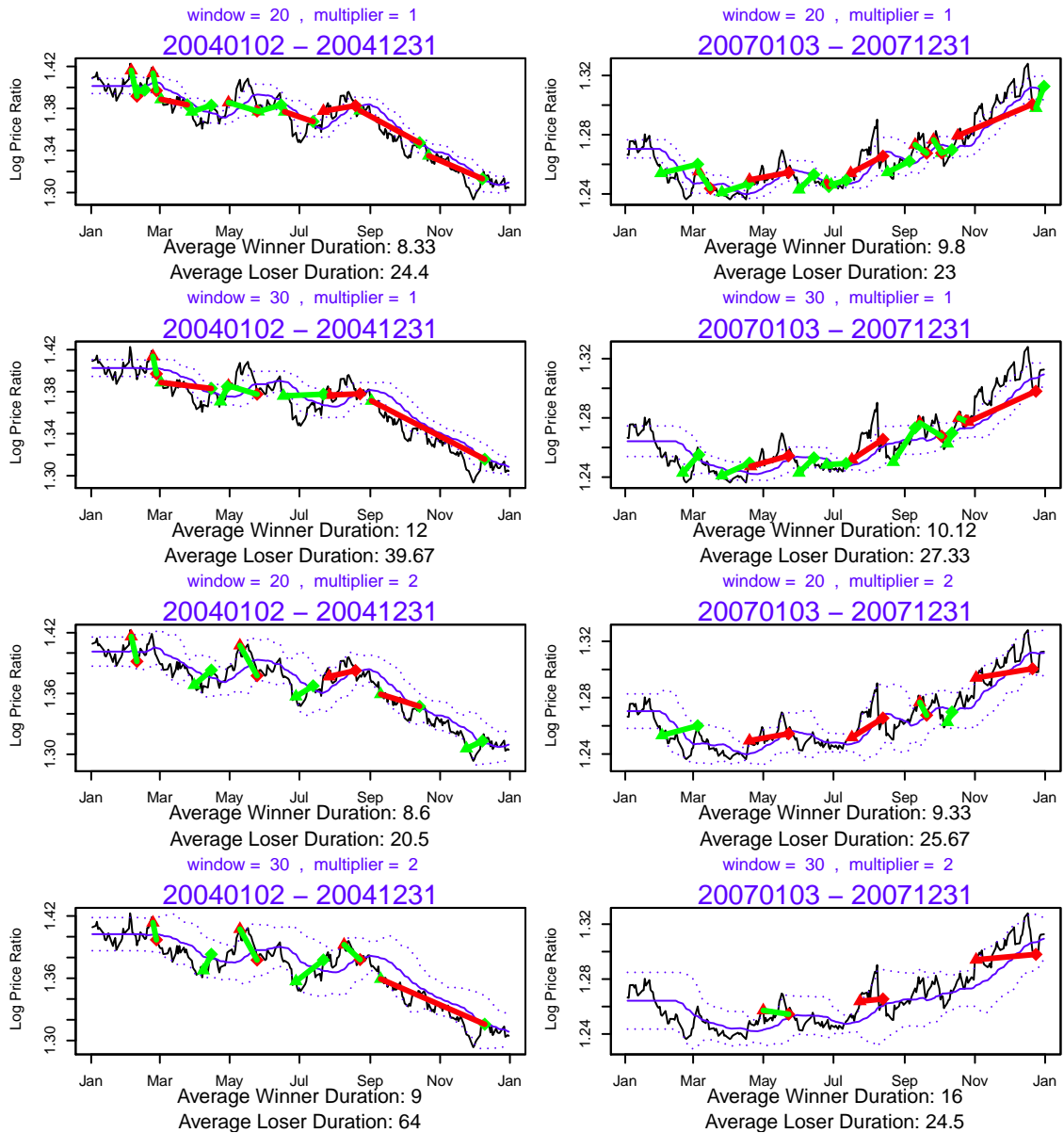} %%%%%%%%%%% HERE HERE HERE LOOK AT HERE...
\caption{The results of the Bollinger Band pairs trading strategy of the SAP versus the Nikkei simulated over different time periods using various values of the
parameters n and k. The average durations indicate that the nature of the Bollinger Band methodology is such that winning trades have an average 
duration consistently shorter than that of losing trades. 
}\label{FIG:all_pairs}

\end{figure}
%%====================================================================================================================================================================
% WEIGHTED AGE EXAMPLE
\clearpage
\section{An Illustration of the Weighted Age Relation $\lambda = \frac{2}{n+1}$} 
\begin{figure}[tbh] 
\centering\includegraphics[scale=0.9]{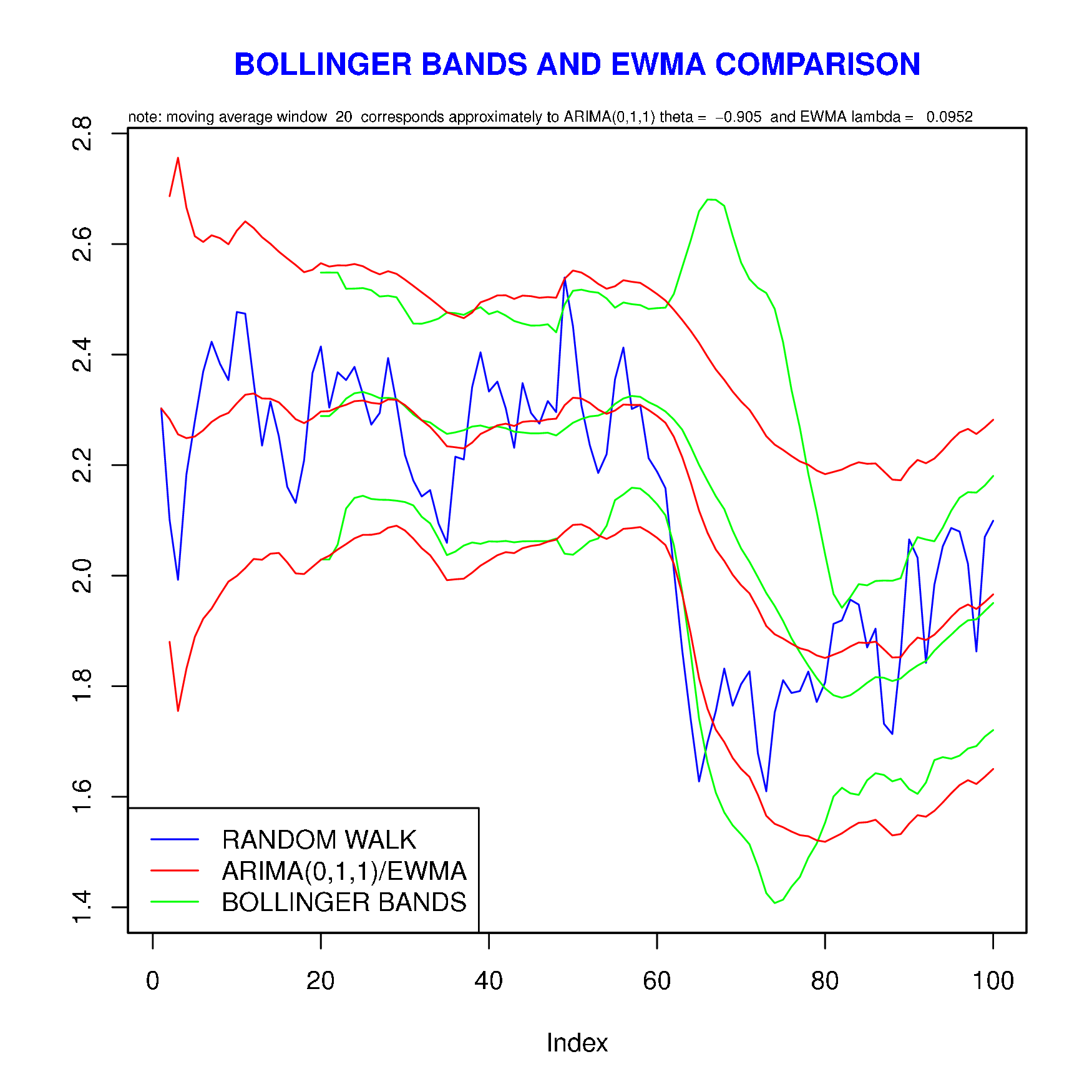}
\caption{The figure shows how similar Bollinger Bands are to the EWMA when the weighted age is matched using the relation $\lambda = \frac{2}{n + 1}$
The approximation is quite reasonable with respect to the center line but not quite as close with respect to $BBUpper$ and $BBLower$}. 
\label{FIG:ewma_boll}
\end{figure}
%%====================================================================================================================================================================
%% BEGINNING  OF PROOF
\clearpage
\section{A Proof That Any Trade In A Bollinger Band Pairs Trading Strategy Has a Non-Negative Return If and Only If The Total Duration Is Less Than 
Or Equal To n Where n Is The Rolling Window Size.}

\noindent Before going into the details of the theorem, we should point out that the theorem is applicable to a larger class of models than just the BBPT strategy. 
Since the theorem result is independent of the band width multiplier, $k$, it is applicable to any trading algorithm where the log
price  of the traded asset 
being some threshold distance away from its moving average triggers the entry signal and the subsequent crossing back of the log price of the traded asset through 
the moving average triggers the exit signal. A well known example of such a strategy is the BBPT strategy but there may be other propietary moving 
average type strategies that meet this criterion also. One obvious example would be where Bollinger Bands are used on a single stock itself rather 
than a pair of stocks. In that case, $y_{t}$ would denote the price itself rather than the price ratio but the result of the theorem would still apply.

\noindent Since Theorem~\ref{THM:theorem1b} requires nine lemmas before it can be proven, details about the notation used and assumptions made are provided below.
The two figures that then follow make the assumptions and notation described below more tangible. Figure \ref{FIG:plot_0} on page \pageref{FIG:plot_0} 
illustrates how the endpoints of the moving average duration are defined. Figure \ref{FIG:plot_0b} on page \pageref{FIG:plot_0b} uses a long trade as an 
example to illustrate other assumptions.

\begin{itemize}
\item[1.] $y_{t}$ denotes the price ratio of the paired asset at time $t$ and the term ``Bollinger Band exit rule'' refers to the rule where one exits
from the trade when the $log(priceratio)$ at time $t$ crosses through or is equal to the moving average of the $log(priceratio)$ at time $t$.

\item[2.] $t^{*}$ denotes the entry time of a trade and $t^{**}$ denotes the time period associated with a trade duration equal to the window size.   
For example, if the window size $n$ was equal to $10$ and a trade was entered into at $t^{*} = 15$, then $t^{**} = t^{*} + n - 1 = 24$.  Therefore, a trade entered
into at the beginning of $t=t^{*}$ and exited from at the end of $t=t^{**}$ would be viewed as having trade duration of $10$ periods. 

\item[3.] We assume a one period time lag between the entry signal time and trade entry time purely for clarity purposes.
We assume that slippage\footnote{Slippage during entry is defined as price erosion due to the delay between when a trade entry is signaled and when the
trade is actually entered. Slippage during exit is defined analogously.} does not occur during the one unit time period between the entry signal and the 
trade entry. In fact, with respect to entry slippage, we go further than this by assuming that not only is there never price erosion during entry but 
also that there is a non-zero infinitesimal price improvement\footnote{The price improvement on entry assumption would not be necessary if we assumed that 
there was no time lag between 
the signal and the entry. By assuming a one period time lag, we separate the defined window from the signal price and provides more clarity when explaining the steps
of the proof}. For example, if a long entry signal is triggered by a price say $log(y_{t^{*}-1})$, then we will assume that the actual entry occurs at a price 
$log(y_{t^{*}-1}) - \delta$ where $\delta > 0$. This additional price improvement assumption is only needed so that edge cases do not need to be considered in the 
steps of the proof. Also, the size of $\delta$ does not affect the derivation of the result as long as it positive so it can be assumed that $\delta$ is 
infinitesimally small. 

\item[4.] One can think of the discrete time block denoted by $t$ as having a length of one period and an n period window as being a 
set of these $n$ blocks stacked adjacently to each other. The convention associated with a given window with endpoints $t^{*}$ and $t^{**}$ is that 
entries can occur at $t^{*},t^{*}+1,t^{*}+2,t^{*}+3,\ldots, t^{**}$ and exits can occur at $t^{*}+1,t^{*}+2,t^{*}+3,t^{*}+4,\ldots,t^{**}+1$. If we say that 
there was an exit signal at time $t^{\prime}$, implicit in this statement is that the first possible exit time is the end of time $t^{\prime}$ 
which is equivalent to the time right before the beginning of period $t=t^{\prime}+1$. Note that this does not imply the possibility of slippage during exit 
because it is assumed that the price process of $log(y_{t})$ is discrete so that the price does not change during the time block associated with the period 
labeled $t^{\prime}$. Simply speaking, we assume that the discrete price process is such that the value of  $log(y_{t^{\prime}})$ at the beginning of the time block 
denoted by $t^{\prime}$ is the same the value of $log(y_{t^\prime})$ at the end of the same time block $t^{\prime}$.  We then assume an instantaneous change in 
$log(y_{t})$ just as the beginning of the new time block $t^{\prime}+1$ is reached. Obviously, this is an over-simplification of how the actual price process 
evolves but for purposes of clarity we make this assumption.  The main point is that assuming an exit always occurs at the end of a time block is equivalent to 
assuming that there is an instantaneous exit whenever an exit signal is observed and is therefore equivalent to assuming zero slippage on exit.

\item[5.] Although the proof relies on the use of $log(price ratio)$ as the scale on the vertical axis, this is not a restrictive assumption because 
the entries and exits calculated can always be transformed back to price ratio space. The relation that results from using the $log(price ratio)$ on the 
vertical axis will be stated as a lemma. Consequently, before proving the main result, we need to prove this lemma and various other lemmas which will be used in the 
proof of Theorem~\ref{THM:theorem1b}.\footnote{All of the lemmas contain arguments under the assumption that the BBPT strategy generated a long trade. This is without 
loss of generality because symmetric arguments hold if the assumption was that a short trade had been generated.} 
\end{itemize}

\begin{figure}[t!] 
\begin{center}
\includegraphics[scale=1.0]{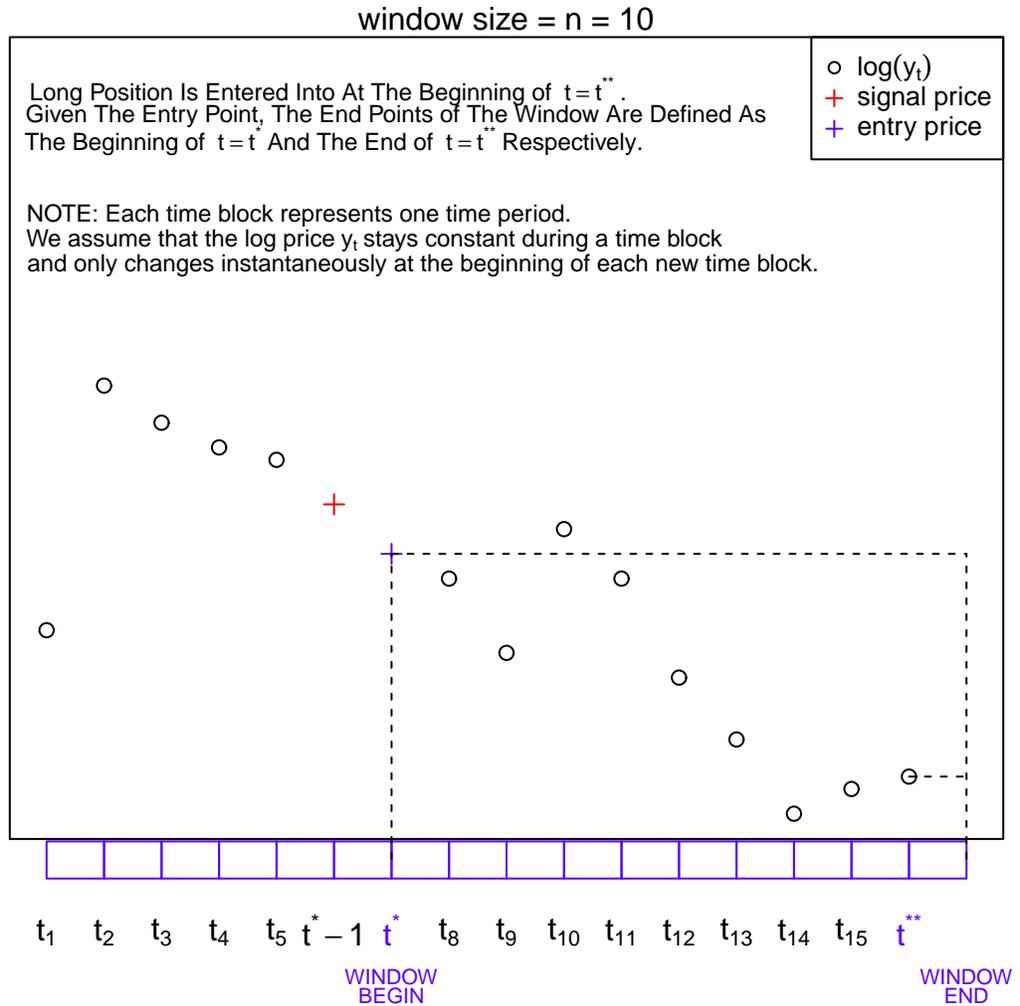}
\end{center}
\caption{ 
Illustrating how the time periods and window duration endpoints are constructed in the proofs of the lemmas and the main theorem. 
A long trade is triggered at $t=t^{*}-1$ and entered into at the beginning of $t=t^{*}$. The first possible exit is at the end of $t=t^{*}$ which
is equivalent to the beginning of $t=t^{*}+1$. The trade duration is equivalent to the rolling window size = 10 when the trade exit occurs at the end of 
$t=t^{**}$ which is equivalent to the beginning of the period $t=t^{**}+1$.}
\label{FIG:plot_0}
\end{figure}
\clearpage
\begin{figure}[h!] 
\begin{center}
\includegraphics[scale=1.0]{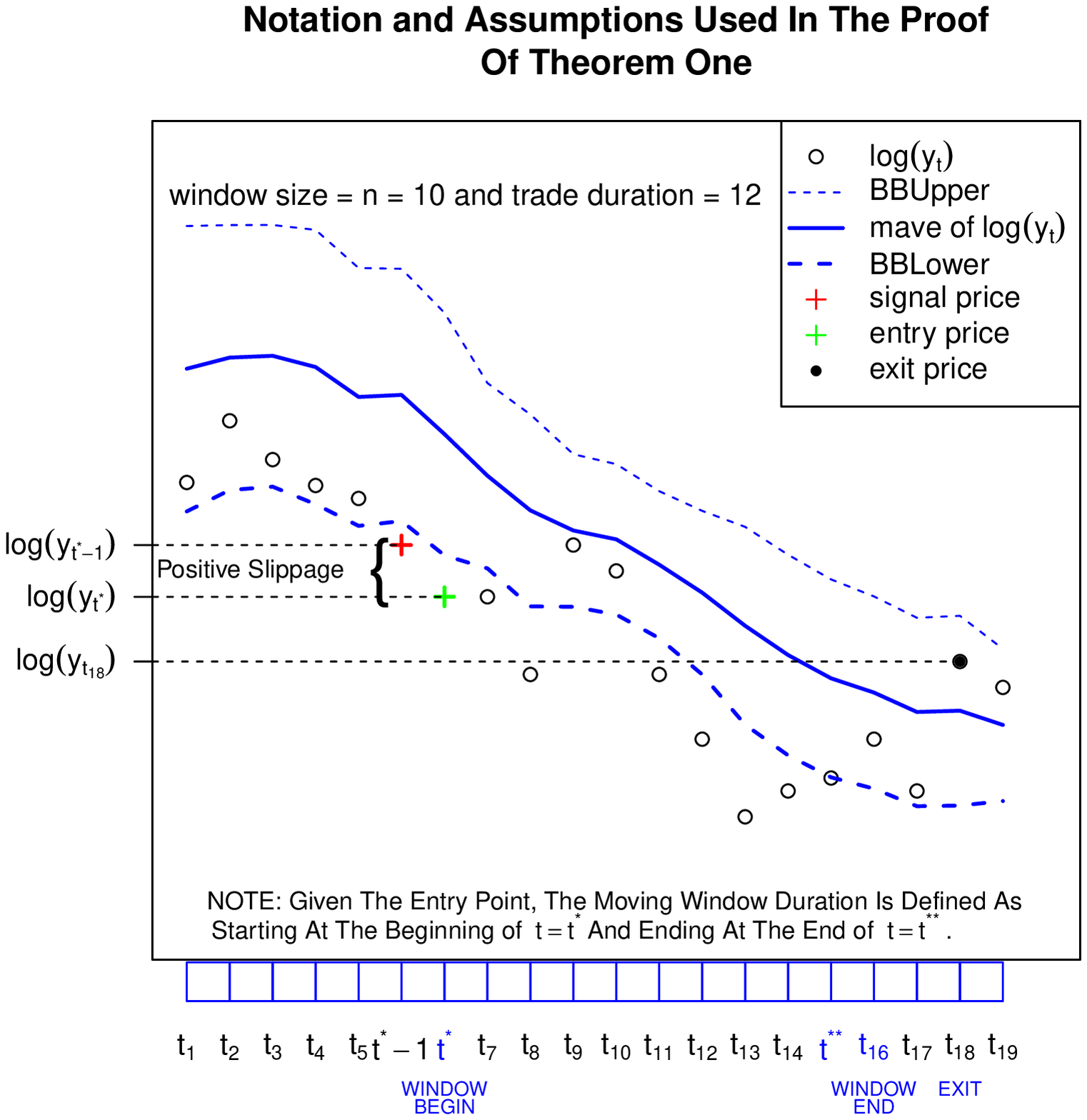}
\end{center}
\caption{ 
Illustrating the various notation and assumptions used in the proofs of the various lemmas and main theorem. A long trade is triggered at $t=t^{*}-1$ and 
entered into at $t=t^{*}$. The total trade duration is longer than the moving average window duration because the actual trade duration = 12 and the moving average 
window duration = 10. The overall log return of the trade is negative.}
\label{FIG:plot_0b}
\end{figure}
\clearpage
\subsection{Preliminary Lemmas}
%%============================================================================================
\begin{lemma}  \label{LEM:lemma1}
Let $y_{t}$ denote the price ratio of a paired asset at time $t$ and consider a window of size n whose endpoints are denoted as $t^{*}$ and $t^{**}$. 
Assume that a long trade has been generated at the beginning of $t=t^{*}-1$ so that the entry takes place at the beginning of $t=t^{*}$. Suppose that 
the Bollinger Band moving exit rule is ignored in that the position is held for a fixed n = rolling window size periods. If the overall log return of the 
paired asset over the period from the beginning of  $t = t^{*}$ to the end of $t = t^{**}$ is $\epsilon$ where $-\infty < \epsilon < \infty$, then the
following relation holds:
\begin{equation*}
log(y_{t^{**}}) = log(y_{t^{*}}) + \epsilon
\end{equation*}
\end{lemma}
\begin{proof}
Consider the interval $[t^{*},t^{**}]$. Since the return over this interval is $\epsilon$, by the additiviy of log returns this implies that  
$\sum_{t=t^{*}+1}^{t=t^{**}} \triangle_{t} log(y_t) = \epsilon$. But the terms in the sum represent a telescoping series which reduces to 
$log(y_{t^{**}}) - log(y_{t^{*}})$ so that $log(y_{t^{**}}) = log(y_{t^{*}}) + \epsilon$. Although the relation is obvious, it will turn out to be useful
when proving various other lemmas that follow. 
\end{proof} \noindent A sample plot of log prices is shown in Figure \ref{FIG:lemma1} below and gives the intuition behind Lemma \ref{LEM:lemma1}.

\begin{figure}[h!] 
\begin{center}
\includegraphics[scale=0.85]{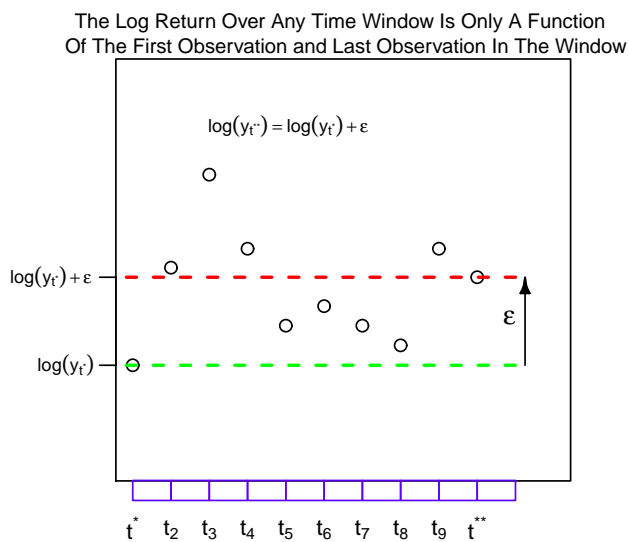}
\end{center}
\caption{ 
Illustrating that only the first price observation and last
price observation are needed to calculate the log return over any time window. }
\label{FIG:lemma1}
\end{figure}
\clearpage
%%=============================================================================================================================================
\begin{lemma}  \label{LEM:lemma2}
Let $y_{t}$ denote the price ratio of a paired asset at time $t$ and consider a window of size n whose endpoints are denoted as $t^{*}$ and $t^{**}$. 
Assume that a long trade has been generated at the beginning of $t=t^{*}-1$ so that the entry takes place at the beginning of $t=t^{*}$. Assume that the Bollinger 
Band exit rule is the exit rule. Then, if $y_{t}$ remains constant from the beginning of  $t=t^{*} $ up until the end of $ t=t^{**}$, 
then $mave_{t^{**}} = log(y_{t^{**}}) = log(y_{t^{*}})$ and the long trade will be exited at the end of $t = t^{**}$.
\end{lemma}
\begin{proof}
Notice that at the end of period $t = t^{**}$, aside from $log(y_{t^{*}})$ itself, $mave_{t^{**}}$ will not contain any of the points contained in the window 
when $t = t^{**}$ is the right endpoint of the window. Therefore only points that are equal to $log(y_{t^{*}})$ will be contained in the calculation of 
$mave_{t^{**}}$. Obviously, the average of the n $log(y_{t})$ values in the window at 
$t^{**} = log(y_{t^{*}})$. Therefore $mave_{t^{**}} = log(y_{t^{*}})$. But $log(y_{t})$ did not change over the time between $t^{*}$ and $t^{**}$ so 
$log(y_{t^{**}}) = log(y_{t^{*}})$ which means $mave_{t^{**}} = log(y_{t^{**}})$ so that a trade exit is triggered. Therefore the trade is exited with a log return 
equal to zero since $log(y_{t})$ was constant over the trade duration. An illustration of this argument is provided in Figure \ref{FIG:lemma2} on 
page \pageref{FIG:lemma2}.
\end{proof}

\begin{figure}[t!] 
\includegraphics[scale=1.0]{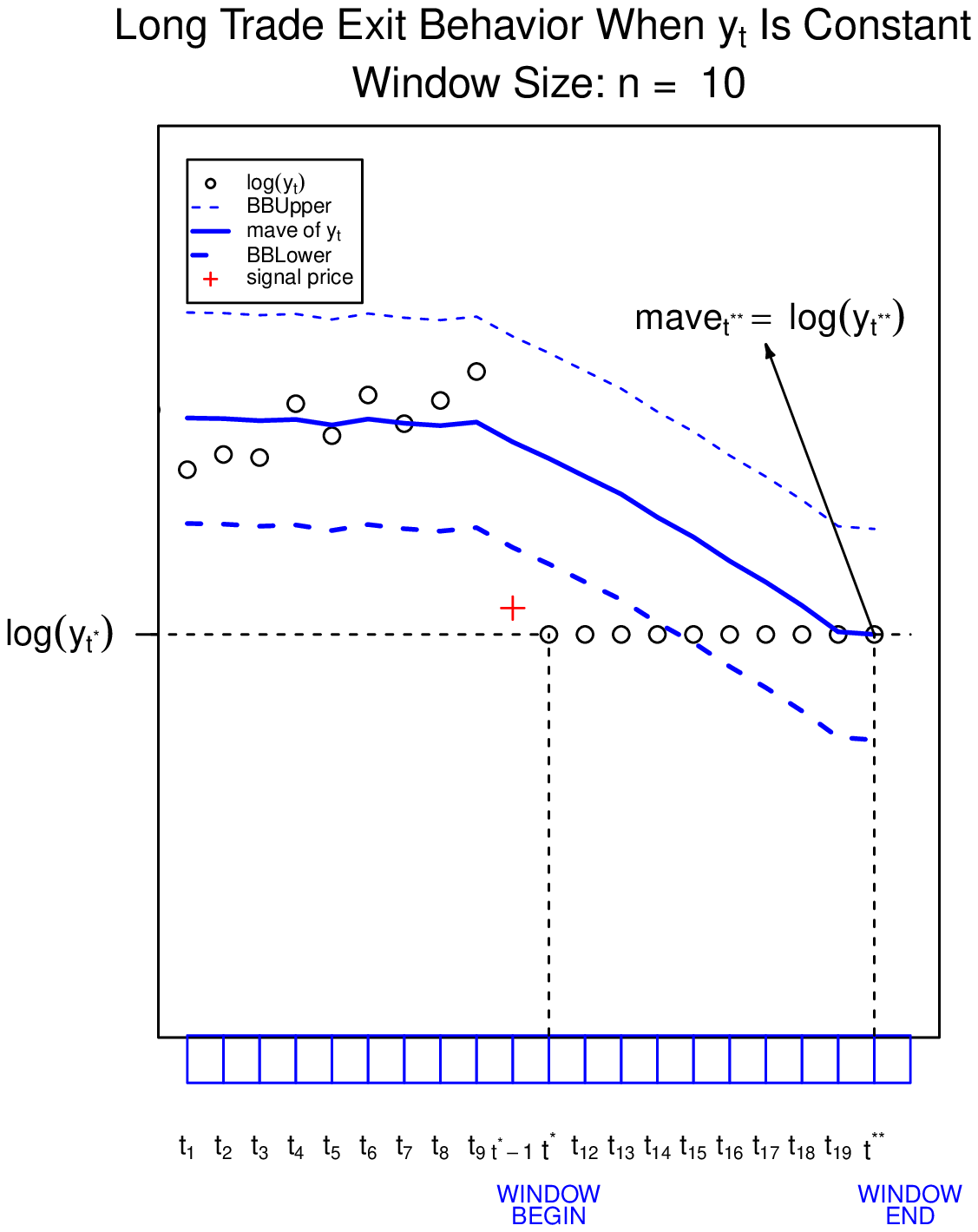}
\caption{Illustrating that a trade will exit
at the end of $t=t^{**}$ when the price
remains constant over n = rolling window size periods.}
\label{FIG:lemma2}
\end{figure}

%%=============================================================================================================================================
\begin{lemma}  \label{LEM:lemma3}
Let $y_{t}$ denote the price ratio of a paired asset at time $t$ and consider a window of size n whose endpoints are denoted as $t^{*}$ and $t^{**}$. 
Assume that a long trade has been generated at the beginning of $t=t^{*}-1$ so that the entry takes place at the beginning of $t=t^{*}$. 
Also, suppose that the Bollinger Band exit rule is ignored in that the position is held for a fixed n = rolling window size periods. 
If the overall log return of the paired asset over the period from the beginning of $t = t^{*}$ to the end of $t = t^{**}$ is $-1.0 \times \epsilon$ 
where $\epsilon > 0$, then the following relation holds:
\begin{equation*}
log(y_{t^{*}}) - \epsilon < mave_{t^{**}} < log(y_{t^{*}})
\end{equation*}
\noindent where $mave_{t^{**}} = \frac{\sum_{t=t^{*}}^{t=t^{**}} log(y_{t})}{n}$
\end{lemma}

\begin{proof}
\noindent 
We  can obtain the upper and lower  bounds for $mave_{t^{**}}$ by considering two extreme scenarios in which the overall log return of the trade is 
$-1.0 \times \epsilon$.\footnote{Notice that this proof is essentially assuming a discrete process for $log(y_{t})$ which is consistent 
with the original assumption that $y_{t}$ process only changes at the end of each period $t$.} This two extreme scenario proof methodology is justified 
because any other realization where the return over the window is equal to $-1.0 \times \epsilon$ is a realization that also maintains the same upper and
lower bounds. 

\noindent First consider scenario one where $log(y_{t})$ moves to the level $(log(y_{t^{*}}) - \epsilon)$ at the beginning of $t=t^{*}+1$ and then remains 
constant after that until the end of $t=t^{**}$ is reached.\footnote{Note that the log return over the window in scenario one is $-1.0 \times \epsilon$.} 
By definition, $mave_{t^{**}} = \frac{log(y_{t^{*}})}{n} +  \frac{\sum_{t=(t^{*} + 1)}^{t=t^{**}}(log(y_{t^{*}} - \epsilon))}{n}$.
But since $log(y_{t}) = (log(y_{t^{*}}-\epsilon))$ at $t=t^{*}+1$ and remains at that level after $t=t^{*}+1$, the previous expression for $mave_{t^{**}}$ reduces 
to $\left(\frac{1}{n} \times log(y_{t^{*}}) + \frac{n-1}{n} \times (log(y_{t^{*}})- \epsilon)\right)$. 
But since $log(y_{t^{*}}) > (log(y_{t^{*}})- \epsilon)$, this implies that $mave_{t^{**}}  = \left(\frac{1}{n} \times log(y_{t^{*}}) + \frac{n-1}{n} \times 
(log(y_{t^{*}})-\epsilon)\right) > log(y_{t^{*}}) - \epsilon$.\footnote{One can think of this opened lower bound as the discrete counterpart of the integral 
$\frac{\int_{t^{*}}^{t^{**}} (log(y_{t}) - \epsilon)\,dt}{t^{**}-t^{*}}$ where $(log(y_{t})-\epsilon)$ is constant over the period $t=t^{*}$ to $t = t^{**}$.} 
Therefore the opened lower bound for $mave_{t^{**}}$ is $log(y_{t^{*}})-\epsilon$. 

\noindent Next consider scenario two where $log(y_{t})$ is constant over the window and then moves to $(log(y_{t^{*}}) - \epsilon)$ just as the
beginning of $t=t^{**}$ is reached. We can again use the discrete analog of an integration argument to show that $mave_{t^{**}} 
< \frac{(n \times log(y_{t^{*}}))}{n} = log(y_{t^{*}})$.\footnote{One can think of this opened upper bound as the discrete counterpart of the integral 
$\frac{\int_{t^{*}}^{t^{**}} (log(y_{t})\,dt}{t^{**}-t^{*}}$ where $log(y_{t})$ is constant over the period $t=t^{*}$ to $t = t^{**}$.} 
Therefore the opened upper bound for $mave_{t^{**}}$ is $log(y_{t^{*}})$. We illustrate the previous scenario arguments in 
Figure \ref{FIG:lemma3} above.

\end{proof}
\begin{figure}[t!] 
\includegraphics[scale=0.80]{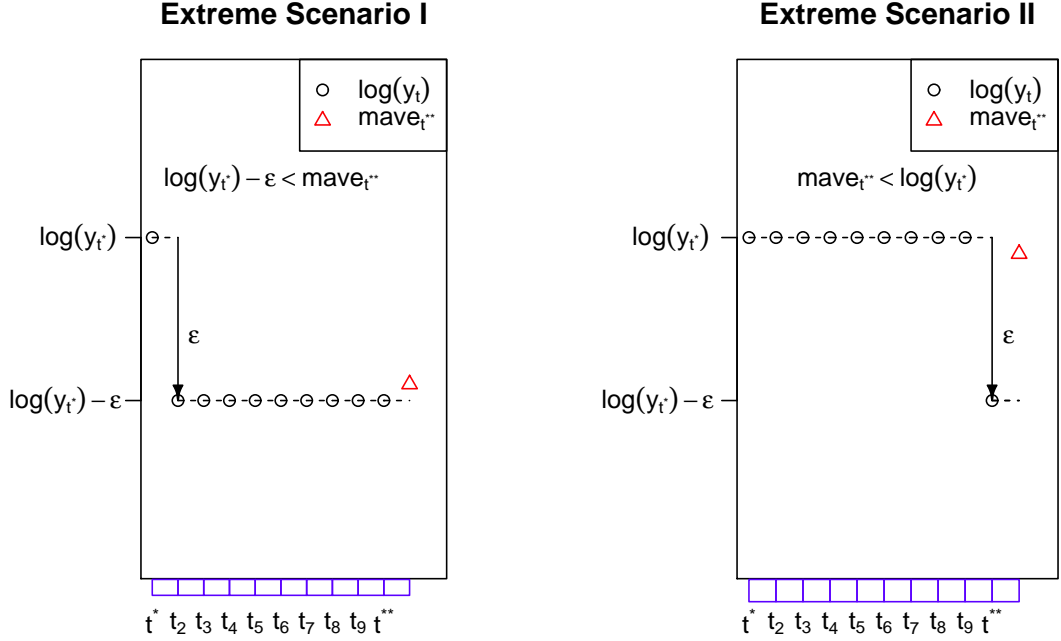}
\caption{Illustrating the two extreme scenarios used
in the proof of Lemma~\ref{LEM:lemma3}:
(a) the LHS inequality and (b) the RHS inequality.}
\label{FIG:lemma3}
\end{figure}

\noindent It will turn out to be convenient to re-write the relation in Lemma~\ref{LEM:lemma3} so that the right hand side is 
an equality. Therefore, we re-state Lemma~\ref{LEM:lemma3} in the following equivalent way:
\begin{equation} \label{EQN:minusbounds}
log(y_{t^{*}}) - \epsilon_{1} < mave_{t^{**}} = log(y_{t^{*}}) - \epsilon_{2}
\end{equation}
\noindent where $mave_{t^{**}} = \frac{\sum_{t=t^{*}}^{t=t^{**}} log(y_{t})}{n}$, $\epsilon_{1} > 0$, $\epsilon_{2} > 0$ and $\epsilon_{2} <  \epsilon_{1}$.

\noindent  In the steps of the proof of Lemma~\ref{LEM:lemma3}, since the rolling window size = n periods was used as the 
holding period, it was unnecessary to consider the values of $log(y_{t})$ for $t < t^{*}$ because, by the time $t=t^{**}$, these values
are not contained in the rolling window with right endpoint $t = t^{**}$. Fortunately, if we make one extra assumption about the values of $log(y_{t})$ 
for a specific period $[t,t^{*})$, then the lower bound in Lemma \ref{LEM:lemma3} can be strengthened to hold for any $t = t^{\prime}$ where 
$t^{*} < t^{\prime} < t^{**}$, rather than just $t^{**}$, resulting in Lemma \ref{LEM:lemma4}

%%=============================================================================================================================================
\begin{lemma}  \label{LEM:lemma4}
\noindent Let $y_{t}$ denote the price ratio of a paired asset at time $t$ and consider a window of size n whose endpoints are denoted as 
$t^{*}$ and $t^{**}$. Assume that a long trade has been generated at the beginning of $t=t^{*}-1$ so that the entry takes place at the beginning of $t=t^{*}$.
Let $t=t^{\prime}$ be any time point between $t=t^{*}$ and $t=t^{**}$ such that $t^{*} < t^{\prime} < t^{**}$. Again suppose that the Bollinger Band moving 
average exit rule is disregarded in that the position is held for a fixed $(t^{\prime} - t^{*}) = n^{\prime}$ periods.  
If the overall log return of the paired asset over a period from the beginning of $t = t^{*}$ to the end of $t = t^{\prime}$ is $-1.0 \times \epsilon$
where $\epsilon > 0$ and  $log(y_{t}) > (log(y_{t^{*}} - \epsilon) ~\forall t$ such that $(t^{*} -n^{\prime}-1) \le t < t^{*}$, then the following 
relation holds:

\begin{equation} \label{EQN:minusboundsgen}
log(y_{t^{*}}) - \epsilon < mave_{t^{\prime}} 
\end{equation}
\noindent where $mave_{t^{\prime}} = \frac{\sum_{t=(t^{*}-n^{\prime}-1)}^{t=t^{\prime}} log(y_{t})}{n}$ 
\end{lemma}

\begin{proof}
\noindent 
We can use the same integration argument used for the lower bound result of Lemma $\ref{LEM:lemma3}$ except, in this case, the  integral used to derive the 
lower bound will now contain the upper limit $t^{\prime}$ rather than $t^{**}$. Note though that, since we are no longer assuming that the return is
calculated over the window with $t=t^{**}$ as the right end point, $mave_{t^{\prime}}$ will still contain values of 
$log(y_{t}) ~\forall t ~(t^{*} - n^{\prime}-1) \le t < t^{*}$. Therefore, the extra condition on $log(y_{t})$  in $[(t^{*} - n^{\prime} -1),t^{*})$ is required in 
order to ensure that the same integration argument will still hold for the lower bound. This is because the integration starts from $t=t^{*}$. Therefore,
for the relation to be true when part of the interval is to the left of the window and therefore not included in the integral, the extra condition is required 
for the $log(y_{t})$ values in that part of the interval.
\end{proof}

%%============================================================================================
\begin{lemma}  \label{LEM:lemma5}
Let $y_{t}$ denote the price ratio of a paired asset at time $t$ and consider a window of size n whose endpoints are denoted as $t^{*}$ and $t^{**}$. 
Assume that a long trade has been generated at the beginning of $t=t^{*}-1$ so that the entry takes place at the beginning of $t=t^{*}$. 
Again suppose that the Bollinger Band moving average exit rule is disregarded in that the position is held for a fixed n = rolling window size periods. 
If the overall log return of the paired asset over a period from the beginning of $t = t^{*}$ to the end of $t = t^{**}$ is $ +1.0 \times \epsilon$ where 
$\epsilon > 0$, then the following relation holds:
\begin{equation*}
log(y_{t^{*}}) < mave_{t^{**}} < log(y_{t^{*}}) + \epsilon
\end{equation*}
\noindent where $mave_{t^{**}} = \frac{\sum_{t=t^{*}}^{t=t^{**}} log(y_{t})}{n}$
\end{lemma}
\begin{proof}
The proof uses similar integration arguments to those used in Lemma~\ref{LEM:lemma3} so the details will not be included here. 
We illustrate the previous scenario arguments in Figure \ref{FIG:lemma5} above.

\end{proof}
\begin{figure}[t!] 
\includegraphics[scale=0.80]{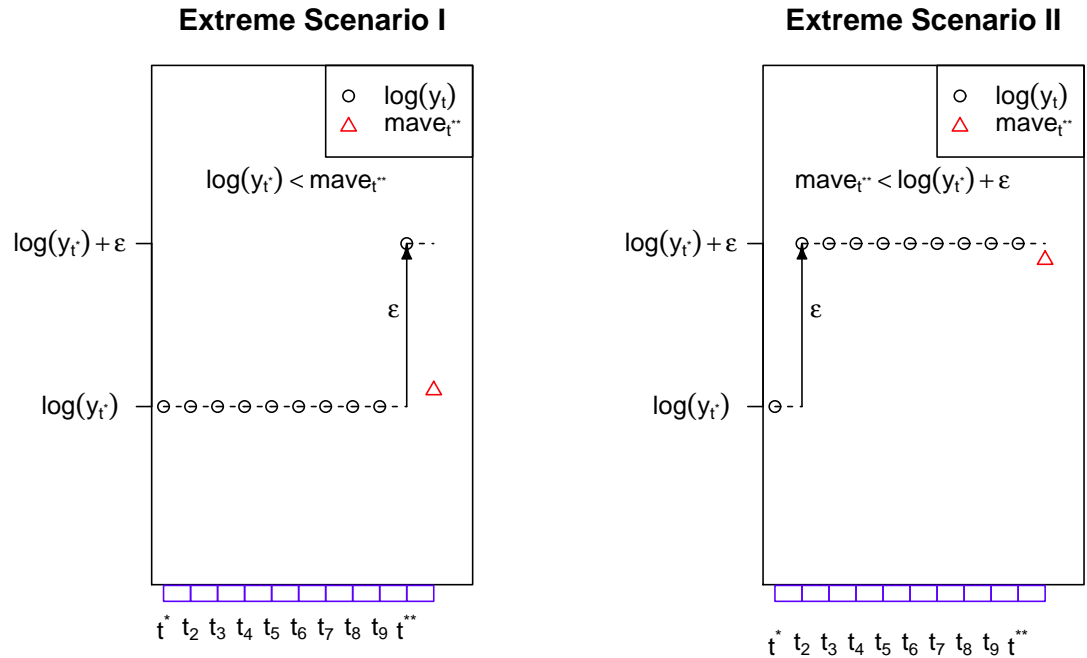}
\caption{Illustrating the two extreme scenarios used
in the proof of Lemma~\ref{LEM:lemma5}:
(a) the LHS inequality and (b) the RHS inequality.}
\label{FIG:lemma5}
\end{figure}

\noindent Just as was the case with Lemma~\ref{LEM:lemma3}, since the rolling window size = n periods was used as the 
holding period in the proof of Lemma~\ref{LEM:lemma5}, it was unnecessary to consider the values of $log(y_{t})$ for $t < t^{*}$ in the proof because, by the 
time $t=t^{**}$ is reached, these values are not contained in the rolling window with right endpoint $t = t^{**}$. 
Again, if we make one extra assumption about these values, then the upper bound in Lemma \ref{LEM:lemma5} can be strengthened to 
hold for any $t = t ^{\prime}$ where $t^{*} < t^{\prime} < t^{**}$ rather than just $t^{**}$, resulting in Lemma \ref{LEM:lemma6}.
%%============================================================================================
\newpage
\begin{lemma}  \label{LEM:lemma6}
Let $y_{t}$ denote the price ratio of a paired asset at time $t$ and consider a window of size n whose endpoints are denoted as 
$t^{*}$ and $t^{**}$.  Assume that a long trade has been generated at the beginning of $t=t^{*}-1$ so that the entry takes place at the beginning of $t=t^{*}$.
Let $t=t^{\prime}$ be any time point between $t=t^{*}$ and $t=t^{**}$ such that $t^{*} < t^{\prime} < t^{**}$. Again suppose that the Bollinger Band moving 
average exit rule is disregarded in that the position is held for a fixed $(t^{\prime} - t^{*}) = n^{\prime}$ periods.  
If the overall log return of the paired asset over the period from the beginning of $t = t^{*}$ to the end of $t = t^{\prime}$ is $+1.0 \times \epsilon$
where $\epsilon > 0$ and  $log(y_{t}) < mave_{t^{*}} ~\forall t$ such that $(t^{*} -n^{\prime}-1) \le t < t^{*}$, then the following 
relation holds:
\begin{equation} \label{EQN:plusboundsgen}
mave_{t^{\prime}} < log(y_{t^{*}}) + \epsilon
\end{equation}
\noindent where $mave_{t^{\prime}} = \frac{\sum_{t=(t^{*}-n^{\prime}-1)}^{t=t^{\prime}} log(y_{t})}{n}$
\end{lemma}

\begin{proof}
\noindent 
We can use the same integration argument used for the upper bound result of Lemma \ref{LEM:lemma5} except, in this case, the integral used to derive the upper bound
will contain the upper limit $t^{\prime}$ rather than $t^{**}$. Note though just as was the case in Lemma \ref{LEM:lemma4}, since we are no longer assuming that 
the return is calculated over the window with $t=t^{**}$ as the right end point, $mave_{t^{\prime}}$ will still contain values of $log(y_{t}) 
~\forall t ~(t^{*} - n^{\prime}-1) \le t < t^{*}$. Therefore, the extra condition on $log(y_{t})$  in $[(t^{*} - n^{\prime} -1),t^{*})$ is required in order to 
ensure that the same integration argument will still hold for the upper bound. Since the integration starts from $t=t^{*}$, for the relation to be true when part 
of the interval is to the left of $t=t^{*}$ and therefore not included in the integral, the extra condition is required for the $log(y_{t})$ values in that part 
of the interval.

\end{proof}

%%==========================================================================================================================================

\noindent The next lemma is stated below.
\begin{lemma} \label{LEM:lemma7}
Let $y_{t}$ denote the price ratio of a paired asset at time $t$ and consider a window of size n whose endpoints are denoted as $t^{*}$ and $t^{**}$. 
Assume that a long trade has been generated at the beginning of $t=t^{*}-1$ so that the entry takes place at the beginning of $t=t^{*}$.  
Let $mave_{t^{*}}$ denote the moving average of the paired asset at time $t^{*}$. Then the maximum possible overall log return that can be generated by the 
trade using the Bollinger Band rule is less than $(mave_{t^{*}} - log(y_{t^{*}})) = \epsilon_{1} > 0$ which is the initial difference between $mave_{t}$ at entry
and $log(y_{t})$ at entry. 
\end{lemma}
\begin{proof} 
We will assume that, with the Bollinger Band exit rule ignored, the aforementioned trade generates an overall log return of $+1.0* \epsilon_{1}$ from the beginning 
of $t=t^{*}$ to end of $t=t^{\prime}$ where $\epsilon_{1} > 0$. Also, we will assume that $t^{\prime} - t^{*} = n^{\prime}$ periods. Then we will show that if one had 
used the Bollinger Band exit rule to exit from the same trade, the overall log return generated would be less than $\epsilon_{1}$. This will complete the proof 
because the overall log return assumed, $\epsilon_{1}$, is equal to the difference between the moving average at entry and the log price at entry.

\noindent So assume that the long trade position that was initiated at $t=t^{*}$ was held for $n^{\prime}$ periods to the end of $t=t^{\prime}$ without 
regard to the Bollinger Band exit rule and that it generated a return of $+1.0 \times \epsilon_{1} = mave_{t^{*}} - log(y_{t^{*}})$. 
Notice that the conditions of Lemma \ref{LEM:lemma6} are met because, since a long trade was generated at $t=t^{*}-1$, it should be the case that 
$log(y_{t}) < mave_{t^{*}} ~\forall t$ such that $(t^{*} -n^{\prime}-1) \le t < t^{*}$\footnote{In order to be absolutely certain that $log(y_{t^{*}})$ is less 
than $mave_{t^{*}} ~\forall t$ such that $(t^{*} -n^{\prime}-1) \le t < t^{*}$, we need to assume that a separate long trade was not completed during these 
$(n^{\prime} + 1)$ time periods. If a separate long trade was completed during this time period and this trade was generated by a sudden large and sharp downward 
spike in $log(priceratio)$ and exited due to another sudden large and sharp upward spike in $log(priceratio)$, then it possible that the condition will not hold. 
Although the probability of this event is quite small, for this reason we need to make the assumption that a separate long trade was not completed during the 
previous $(n^{\prime}+1)$ periods.}. Therefore, by Lemma \ref{LEM:lemma6}, we know that $mave_{t^{\prime}} < log(y_{t^{*}}) + \epsilon_{1}$ for any $t^{\prime} > t^{*}$ and 
$t^{\prime} <= t^{**}$. But by definition, $log(y_{t^{\prime}}) = log(y_{t^{*}}) + \epsilon_{1}$ which means that $mave_{t^{\prime}} < log(y_{t^{\prime}})$. But this means 
that $mave_{t^{\prime}}$ must have decreased from its original value of $mave_{t^{*}}$ because otherwise it would be equal to $log(y_{t^{\prime}})$ since $log(y_{t^{*}}) 
+ \epsilon_{1} = mave_{t^{*}}$. But if $mave_{t^{\prime}}$ decreased from its original value of $mave_{t^{*}}$, then this implies that $log(y_{t})$ had to have crossed 
it at some earlier period $t^{\prime\prime} < t^{\prime}$ and, since $log(y_{t})$ increased from the beginning of $t=t^{*}$ to the end of $t=t^{\prime}$, the amount that 
$log(y_{t})$ had to increase in order to cross through $mave_{t}$ had to be less than $(mave_{t^{*}} - log(y_{t^{*}})) = \epsilon_{1}$.  Since $t^{\prime}$ was 
arbitrary, this result is true for any $t$ where $t^{*} < t <= t^{**}$, so we have shown that when using the Bollinger Band exit rule, the overall log return 
generated by any trade is less than $(mave_{t^{*}} -log(y_{t^{*}})) = \epsilon_{1}$. An illustration of this argument is provided in Figure \ref{FIG:lemma7} on 
page \pageref{FIG:lemma7}.
\end{proof}

%%============================================================================================
\noindent Finally we need to state and prove Lemma~\ref{LEM:lemma8}.
\begin{lemma} \label{LEM:lemma8}
Let $y_{t}$ denote the price ratio of a paired asset at time $t$ and consider a window of size n whose endpoints are denoted as $t^{*}$ and $t^{**}$. 
Assume that a long trade has been generated at the beginning of $t=t^{*}-1$ so that the entry takes place at the beginning of $t=t^{*}$. Assume that
the Bollinger Band exit rule is being used. Then, any long trade with an overall non-negative log return that has reached the end of $t^{**}$  will be exited 
at the end of $t=t^{**}$. Conversely, any trade with an overall negative log return that has reached the end of $t=t^{**}$ will not be exited at the end of $t=t^{**}$.
\end{lemma}

\begin{proof} %%[Proof of Lemma~\ref{LEM3}] 
Recall that Lemma~\ref{LEM:lemma2} says that if $log(y_{t})$ is constant over the full window from the beginning of 
$t=t^{*}$ to the end of $t=t^{*}$, then there will be an exit at the end of $t=t^{*}$ because $log(y_{t^{**}})$ and $mave_{t^{**}}$ will be equal. 
We again will use two extreme scenarios along with Lemma \ref{LEM:lemma2} in order to prove Lemma \ref{LEM:lemma8}. Figure \ref{FIG:lemma8} on page 
\pageref{FIG:lemma8} provides graphical representations of the two extreme scenarios

\noindent First consider scenario one where $log(y_{t})$ is constant from the beginning of $t=t^{*}$ to the end of $t=t^{**}-1$ and then increases
an infinitesimally small amount equal to $+1.0 \times \epsilon$ at the beginning of $t=t^{**}$. This implies that
the log return over the full window is $+1.0 \times \epsilon$. Note that for any given unit period increase in $log(y_{t})$, by definition 
the moving average $mave_{t}$ always increases by a smaller amount. This fact along with Lemma \ref{LEM:lemma2} implies that $log(y_{t})$ will cross $mave_{t}$ 
from below at $t=t^{**}$ and the trade will exit at the end of $t=t^{**}$.  

\noindent Next consider scenario two where $log(y_{t})$ is constant from the beginning of $t=t^{*}$ to the end of $t=t^{**}-1$ and then decreases 
an infinitesimally small amount equal to $-1.0 \times \epsilon$ at the beginning of $t=t^{**}$. This implies that
the log return over the full window is $-1.0 \times \epsilon$. Note that for any given decrease in $log(y_{t})$, 
the moving average $mave_{t}$ always increases by a smaller amount in absolute value. This fact along with Lemma \ref{LEM:lemma2} implies that $log(y_{t})$ will 
not cross $mave_{t}$ from below at $t=t^{**}$ and therefore the trade will not exit at the end of $t=t^{**}$. 
\end{proof}

\begin{figure}[t!] 
\includegraphics[scale=1.0]{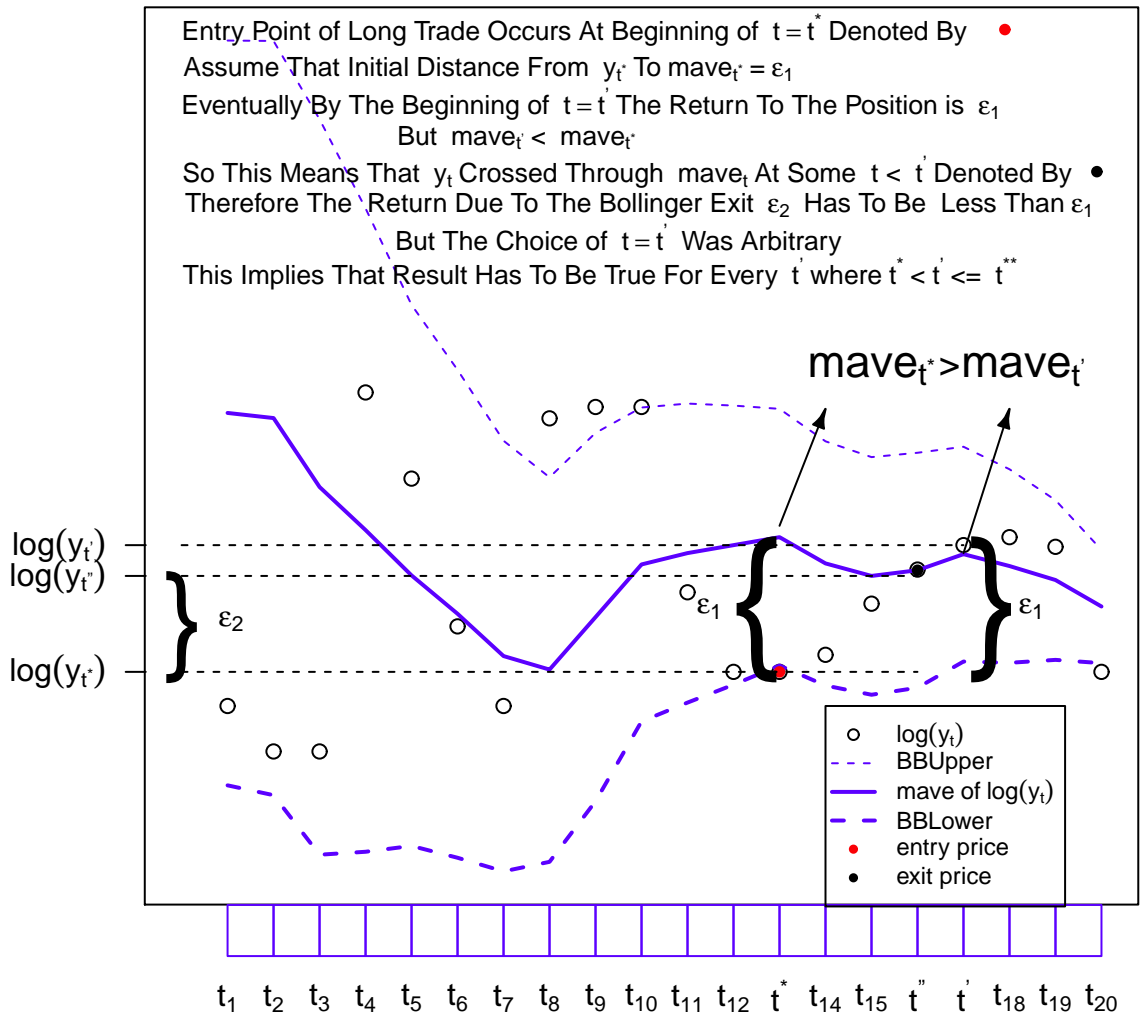}
\caption{Illustrating That The Maximum Log Return of A Bollinger Band Trade Is Always Less Than The Initial Difference Between The Moving Average At Entry And 
The Log(PriceRatio) At Entry.}
\label{FIG:lemma7}
\end{figure}
\clearpage

\begin{figure}[t!] 
\includegraphics[scale=0.8]{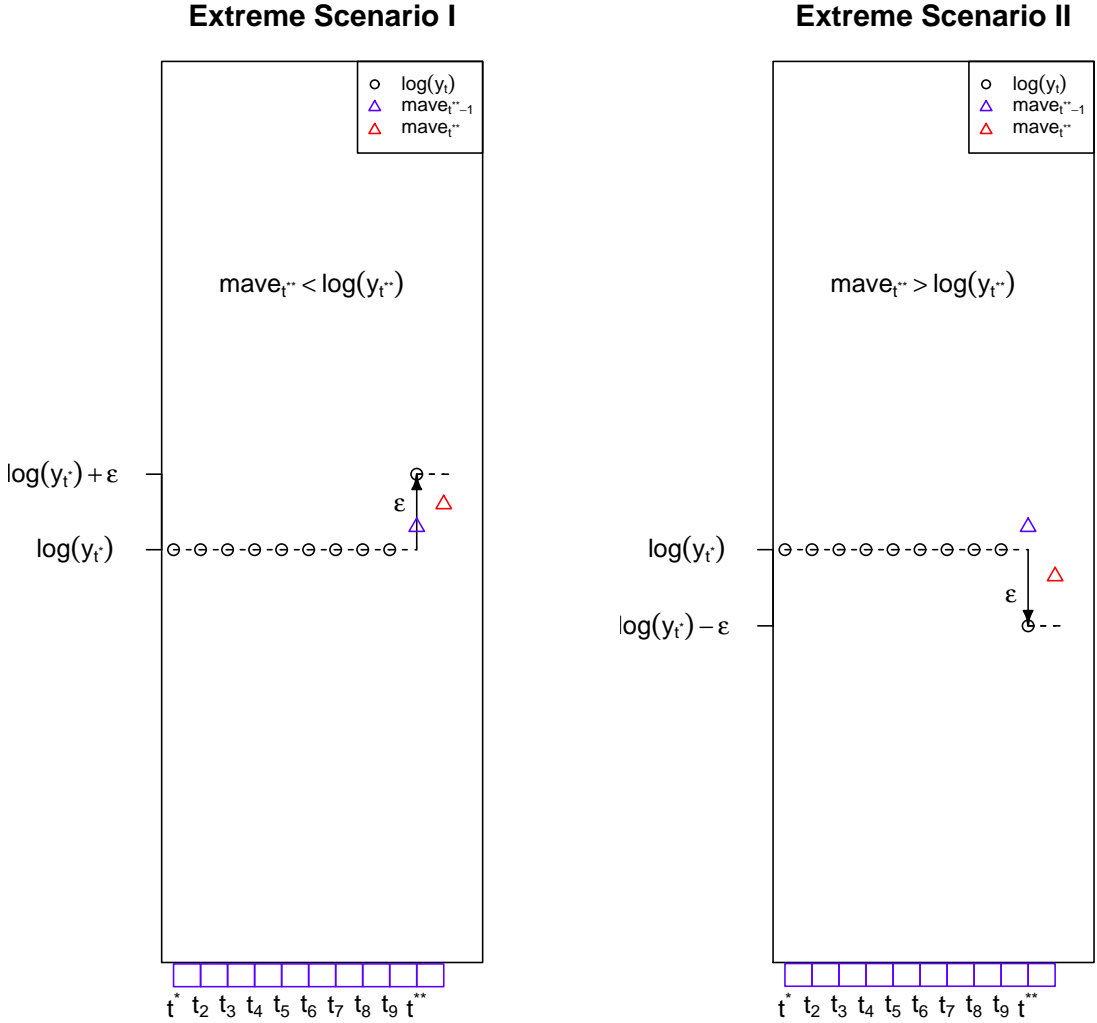}
\caption{Scenario I: An Infinitesimal Positive Return Right Before $t=t^{*}$ Guarantees An Exit At The End Of $t=t^{**}$. 
Scenario II: An Infinitesimal Negative Return Right Before $t=t^{*}$ Guarantees A Non-Exit At The End Of $t=t^{**}$.}
\label{FIG:lemma8}
\end{figure}
%% TO DO
%%================================================
%% NEED PLOT HERE
%%===============================================
\newpage

\noindent Given the various lemmas,  we can prove Theorem~\ref{THM:theorem1b}. We repeat the theorem
statement here.
\newtheorem*{theorem1}{Theorem~\ref{THM:theorem1b}}  %%% MARK See this
\begin{theorem1}  
Assume that the rolling window size in the BBPT strategy  = n,  the band width multiplier = k and that a long trade is generated at  $t = t^{*}-1$. 
so that the entry takes place at the beginning of $t=t^{*}$.  Then the overall log return of this trade using the Bollinger Band exit rule is non-negative if and 
only if the duration of the trade is less than or equal to n; i.e. the trade is exited at a time $t$ less than or equal to $t^{**} = t^{*} + n - 1$. This result is 
independent of the bandwidth multiplier parameter k.
\end{theorem1}

\begin{proof} 
First we prove the if part of Theorem \ref{THM:theorem1b} which means that we need to show that if the pair trade has a non-negative overall log return, 
then the total trade duration has to be less than or equal to $n$ where $n$ is the rolling window size. First, assume that the generated trade is exited at the end
of some time $t = t^{\prime}$. Clearly, if the trade has a non-negative overall log return, then this implies that $log(y_{t^{\prime}}) - log(y_{t^{*}}) >= 0$ 
where $t^{\prime}$ is the exit time of the trade. So let us prove the if part of Theorem \ref{THM:theorem1b} by contradiction: We will assume that 
$log(y_{t^{\prime}}) - log(y_{t^{*}}) >= 0$ (i.e. a non negative overall total log return from entry to exit ) and that $t^{\prime} > t^{**}$ so that the duration 
of the trade is greater than $n$. Then we will show that these assumptions lead to a contradiction.

%%====================================================================================================================

\noindent In order to visualize the argument that follows , a long trade example is  provided in Figure \ref{FIG:plot_ThA} on page 
\pageref{FIG:plot_ThA}. First of all, by assumption, the trade duration is greater than $n$ which means that, at the beginning of $ t=t^{**}$, $mave_{t^{**}}$ 
must have been greater than $log(y_{t^{**}})$ because, if it was not, then based on the Bollinger Band exit rule, the trade would have exited at the end 
of $t=t^{**}$.  Therefore $mave_{t^{**}} > log(y_{t^{**}})$ at the beginning of $t=t^{**}$. Also, by Lemma~\ref{LEM:lemma8}, we know that the total return from the 
beginning of $t=t^{*}$ to the end of $t=t^{**}$ has to be negative because otherwise the trade would have exited at the end of $t=t^{**}$. 
Therefore, we know that $log(y_{t^{**}}) < log(y_{t^{*}})$. So let us assume that the  overall log return from the beginning of $t=t^{*}$ up to the end of 
$t=t^{**}$ is $-1.0 \times \epsilon_{1}$ where $\epsilon_{1} > 0$. Note that, given the latter assumption, equation  (\ref{EQN:minusbounds}) in 
Lemma~\ref{LEM:lemma3} implies that $log(y_{t^{*}}) - \epsilon_{1} < mave_{t^{**}} = log(y_{t^{*}}) - \epsilon_{2}$ where $\epsilon_{2} < \epsilon_{1}$ 
and $\epsilon_{1} > 0$ and $\epsilon_{2} > 0$.

\noindent Now, since we know that the generated trade has not exited by the end of $t=t^{**}$, we can suppose that we are now sitting 
at the beginning of $t = t^{**}$ and can define a new time called the shifted time, $t_{shift}$,  as $t_{shift} = t - ( t^{**} - 1)$ so that the beginning of 
$t_{shift}=1$ corresponds to the beginning of $t=t^{**}$. Now, since the trade has not exited at the end of $t=t^{**}$ and given equation (\ref{EQN:minusbounds}) in
Lemma~\ref{LEM:lemma3}, we can modify our perspective by 
imagining that we are sitting at the beginning of $t_{shift} = 1$ and have just entered a new Bollinger Band trade with the entry point equal to 
the value of $log(y_{t^{**}})$, namely $log(y_{t^{*}}) - \epsilon_{1}$, and the exit point equal to $mave_{t^{**}}$, namely  $log(y_{t^{*}}) - \epsilon_{2}$. 
But, by Lemma~\ref{LEM:lemma7}, the BBPT strategy is such that no trade in BBPT can ever generate more return than the original distance between its entry point, 
$log(y_{t^{**}}) - \epsilon_{1}$, and its initial exit point, $mave_{t^{**}} = log(y_{t^{*}}) - \epsilon_{2}$.
Now, at the beginning of $t_{shift} = 1$, this difference equals $(log(y_{t^{**}}) - \epsilon_{2}) - 
(log(y_{t^{**}} - \epsilon_{1}) = \epsilon_{1} - \epsilon_{2}$. Therefore, an opened upper bound for the log return of the trade going forward from the beginning 
of $t_{shift} = 1$ is $\epsilon_{1} - \epsilon_{2}$. But recall that by assumption, $log(y_{t})$ has decreased by $\epsilon_{1}$ from the beginning of 
$t = t^{*}$ up to the end of time $t = t^{**}$ so the log return of $y_{t}$  during that period is $-1.0 \times \epsilon_{1}$. Therefore, by the additivity of 
log returns, this implies that the maximum possible overall log return of the trade is $ < \epsilon_{1} - \epsilon_{2} - \epsilon_{1} = -1.0*\epsilon_{2}$. But, 
from Lemma~\ref{LEM:lemma3}, $\epsilon_{2} > 0$ so that $-1.0*\epsilon_{2} < 0$. But this means that the trade has to have a negative overall log return which 
is a contradiction because we assumed at the outset of the proof that the trade had a non-negative overall log return. Therefore we have proven the if 
part of Theorem \ref{THM:theorem1b} by contradiction.

\begin{figure}[t!] 
\includegraphics[scale=0.85]{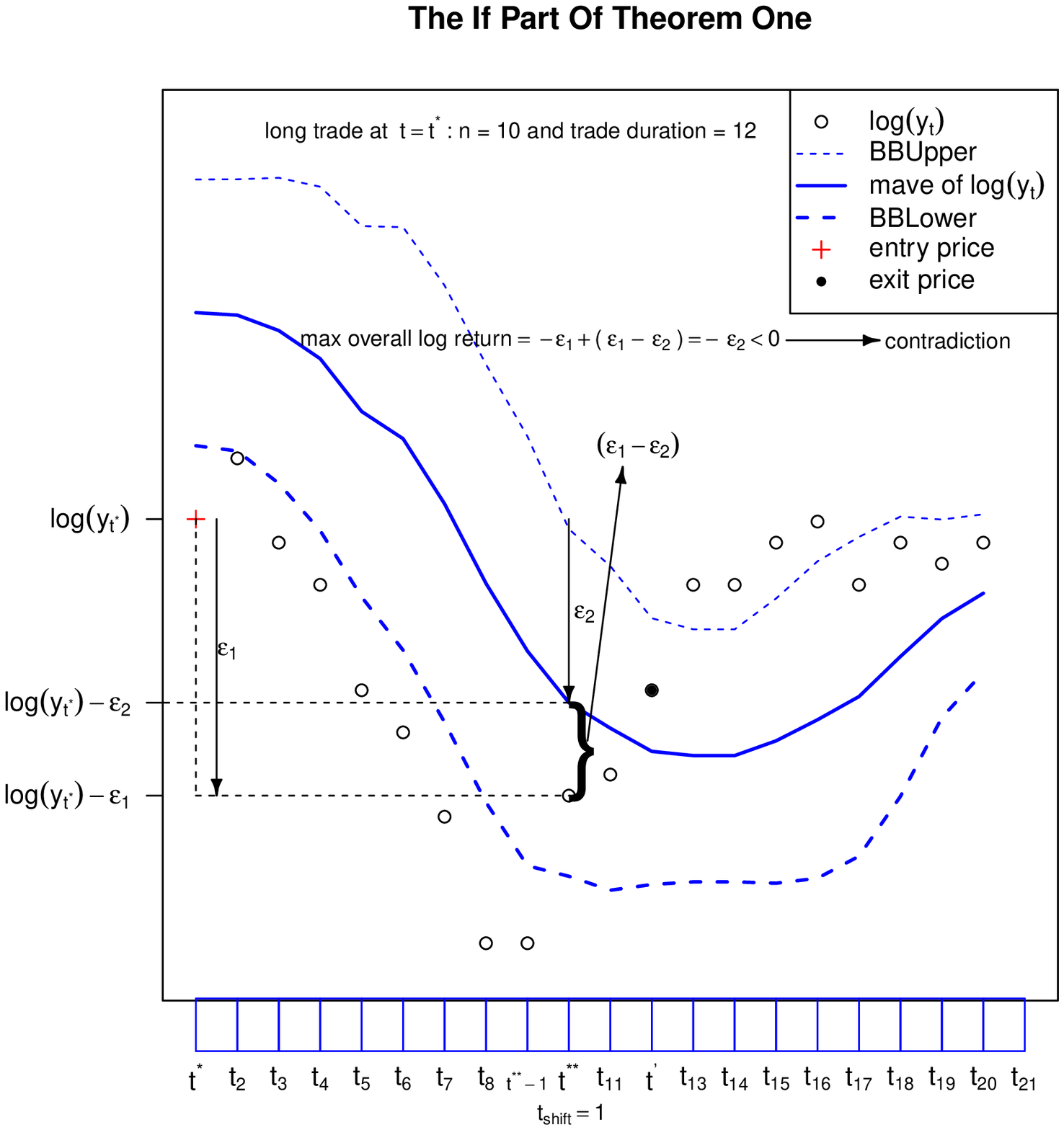}
\caption{Illustrating that a trade
with a non-negative log return cannot have
a duration greater than the rolling window size n.}
\label{FIG:plot_ThA}
\end{figure}

\noindent We still need to prove the only if part of Theorem \ref{THM:theorem1b} which means showing that if the duration of the trade is less than or equal 
to the rolling window size, $n$, then the trade has an overall log return that is non-negative. Again, we will prove the only if part of the theorem by 
contradiction. We will assume that the pair trade duration is less than or equal to $n$ and that the overall log return from the beginning of the entry period 
$t^{*}$ to the end of the exit period $t=t^{\prime}$ is $-1.0 \times \epsilon$ where $\epsilon > 0$ so that the overall log return is negative. Then we will show 
that these assumptions lead to a contradiction. Just as was done with the if part of the theorem, in order that one can visualize the argument that follows, 
a long trade example is provided in Figure \ref{FIG:plot_ThB} on page \pageref{FIG:plot_ThB}. 

%%\clearpage
\noindent First of all, by assumption, the trade duration is less than or equal to $n$ which means that, given the Bollinger Band exit rule, there exists
some $t=t^{\prime} <= t^{**}$ such that at the end of $t = t^{\prime}$, $mave_{t^{\prime}} <= log(y_{t^{\prime}})$. Without loss of generality and so that Figure 
\ref{FIG:plot_ThB} is consistent with the proof, we will assume that $t^{\prime} - t^{*} = n^{\prime} = 4$ so that the exit occurs at the end of 
$t= t^{\prime} = t^{*} + 4$. 

\noindent We need to show that the assumptions above lead to a contradiction. First notice that since the long trade was entered into at $t=t^{*}$, 
this means that the condition $log(y_{t}) > (log(y_{t^{*}} - \epsilon) ~\forall t$ such that $(t^{*} -n^{\prime}-1) 
\le t < t^{*}$ should hold\footnote{Just as was with the case in Lemma \ref{LEM:lemma7}, in order to be certain that the condition holds for all 
$n^{\prime+1}$ periods,  we need to assume  that a separate long trade was not completed in the time period $[(t^{*}-n^{\prime}-1),t^{*}]$.}. This condition 
along with the fact that the overall log return over the interval is $-1.0 \times \epsilon$, allows us to appeal to Lemma \ref{LEM:lemma4} which says that 
$log(y_{t^{*}}) - \epsilon < mave_{t^{\prime}} < log(y_{t^{*}})$. But, by definition, since the total log return over the interval from  
$t=t^{*}$ to $t=t^{\prime}$ is $-1.0 \times \epsilon$, clearly  $log(y_{t^{\prime}}) = log(y_{t^{*}}) - \epsilon$. Therefore it must be the case that 
$log(y_{t^{\prime}}) < mave_{t^{\prime}}$ which means that $log(y_{t})$ could have not crossed through $mave_{t}$ from below at $t=t^{\prime}$. But if 
$log(y_{t^{\prime}})$ did not cross through $mave_{t}$ from below at $t=t^{\prime}$, then this means that there could not have been an exit at $t=t^{\prime}$. 
Therefore we have arrived at a contradiction which completes the proof.

\noindent Both the if and the only if part of Theorem \ref{THM:theorem1b} have been proven so Theorem \ref{THM:theorem1b} has been proven. Any pair 
trade in the BBPT strategy has a non-negative total return if and only if the  duration of the pair trade is less than or equal to $n$ where n is the 
rolling window size.\end{proof}. 
%% END OF PROOF

%% TO DO
%%================================================
%% NEED TO FINISH PLOT BY ADDING EXPLANATION
%%================================================

\begin{figure}[t!] 
\includegraphics[scale=0.90]{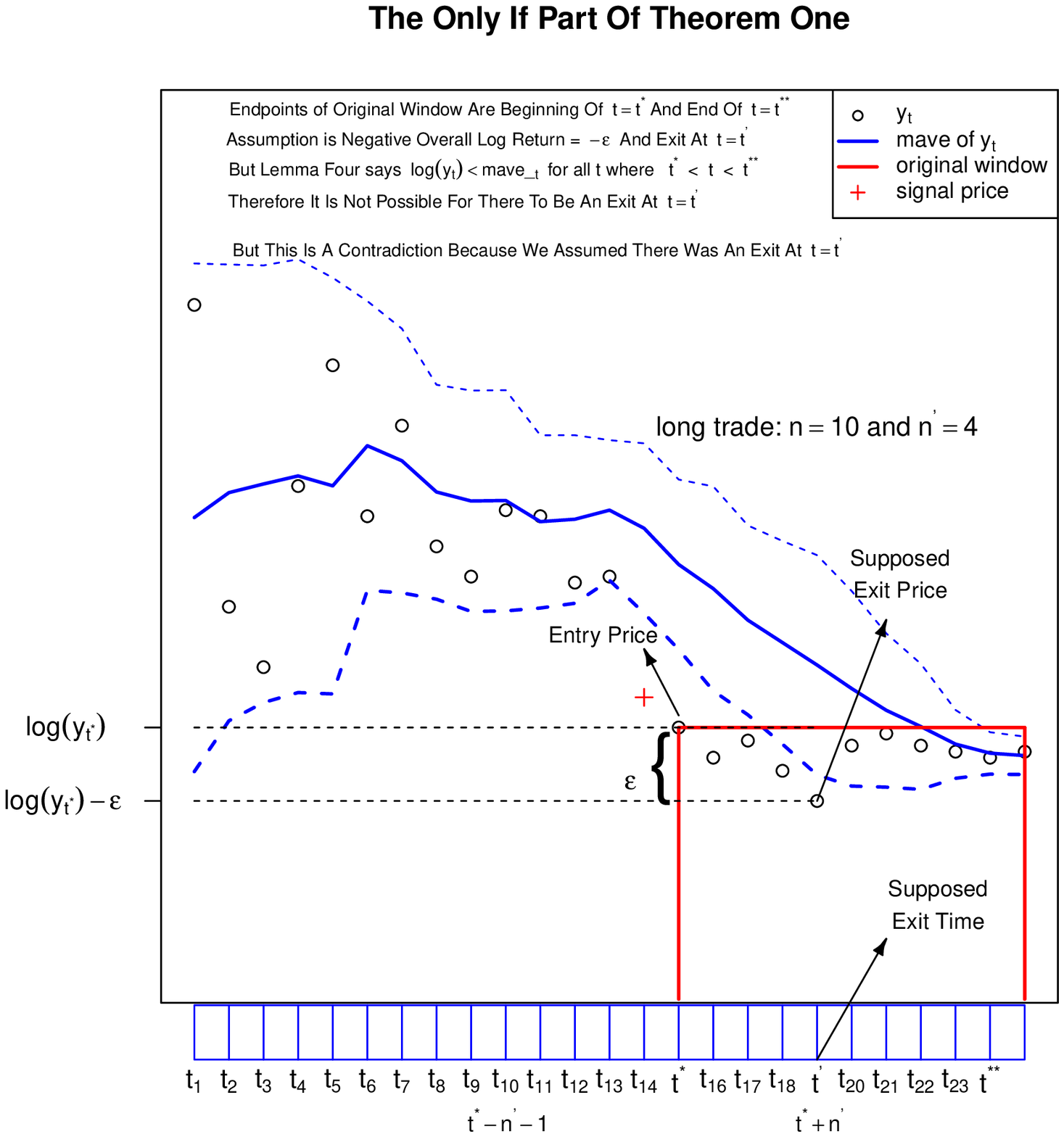}
\caption{Illustrating that a trade
whose duration is less than or equal to the rolling
window size has to have a non-negative overall log 
return.}
\label{FIG:plot_ThB}
\end{figure}

\clearpage
%%====================================================================================================================================================
%% ILLUSTRATING THE TWO PAIRS STRATEGIES
\section{The BBPT Strategy and the Corresponding FFMDPT Strategy} 
\begin{figure}[h!] 
\begin{center}
 \includegraphics[scale=0.80]{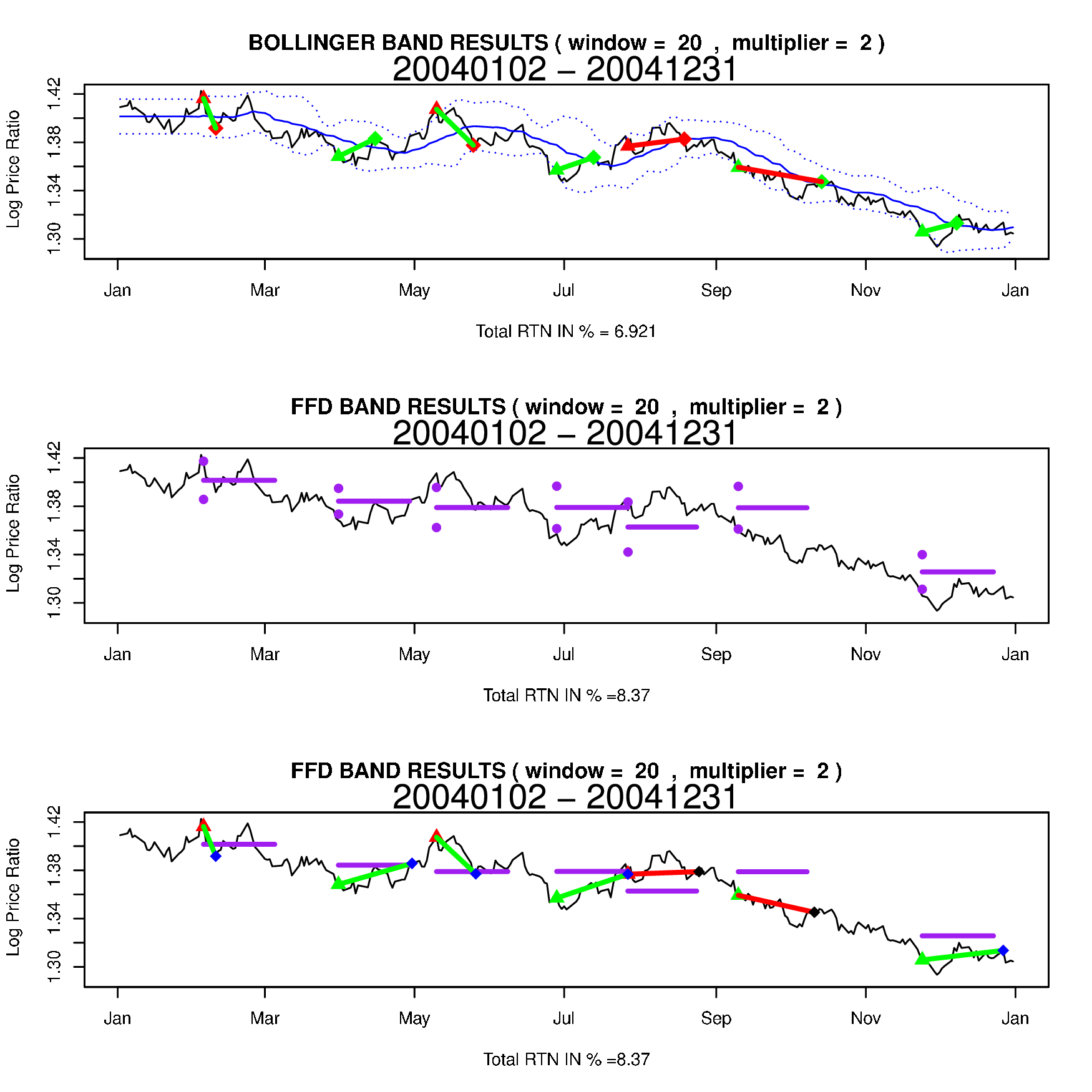}
\end{center}
\caption{ 
The top plot represents the BBPT strategy over 2004 using n = 20 and k = 2 during 2004. The middle and bottom plot illustrate the
the FFDBPT strategy over the same time period. The middle plot excludes the trade line segments for clarity. The purple dots represent $BBUpper$ and 
$BBLower$ at the time of 
entry and the horizontal purple line is the forecast at entry which is constant for n = 20 periods.
The actual trades triggered by the FFMDPT simulation are shown in the bottom FFMDPT plot with a blue triangle at the end
of a line segment indicating that the purple center line was crossed and the black triangle indicating that the maximum 
duration occurred. In the bottom plot, the purple dots at the time of entry are excluded for clarity.  
}\label{FIG:three_plots}
\end{figure}
%% END OF ILLUSTRATING THE TWO PAIRS STRATEGIES
%%========================================================================================================================================================
%% TWO EXAMPLES: IN ONE FFMDPT WINS AND IN THE OTHER IT LOSES
\clearpage
\section{BBPT versus FFMDPT: Two Examples} 
\large\centering \textbf{Example One}
\begin{figure}[h!] 
\begin{center}
 \includegraphics[scale=0.70]{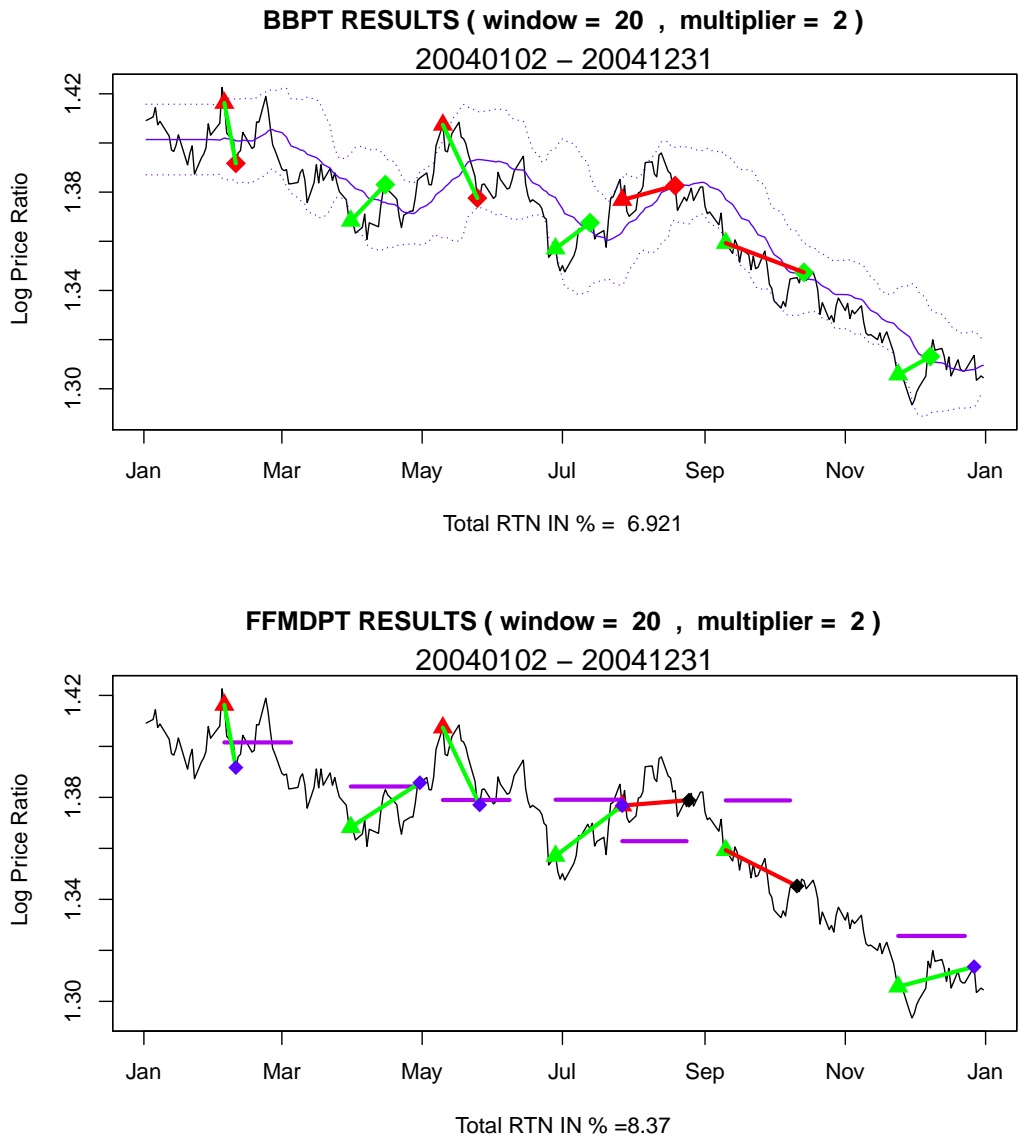}
\end{center}
\caption{ 
A comparison of Bollinger Bands and Fixed Forecast Maximum Duration Bands during 2004 using $n = 20$ and $k = 2$.  The second and third trades in April and June generate 
slightly higher returns in the FFMDPT strategy. Also, the August trade in the BBPT strategy generates a much larger negative return 
compared to the corresponding trade in the FFMDPT strategy.
}\label{FIG:two_plotsa}
\end{figure}
%%========================================================================================================================================================
\clearpage
\large\centering \textbf{Example Two}
\begin{figure}[h!] 
\begin{center}
 \includegraphics[scale=0.70]{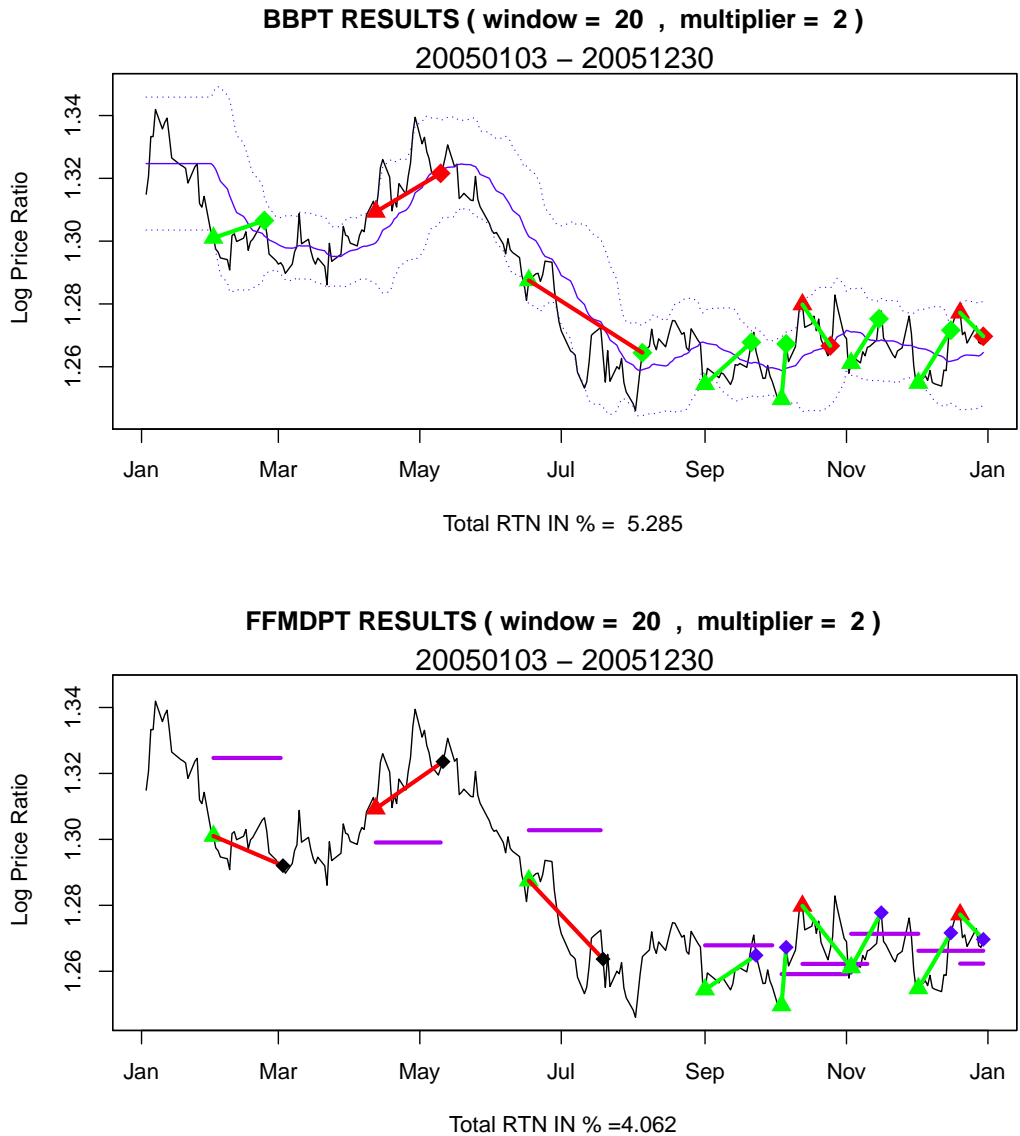}
\end{center}
\caption{ 
A comparison of Bollinger Bands and Fixed Forecast Maximum Duration Bands during 2005 using $n = 20$ and $k = 2$.  
The first trade in early February generates a positive return in the BBPT strategy but a negative return in the FFMDPT strategy. This is because the fixed forecast 
in the FFMDPT strategy is never crossed and extra losses are generated before the exit. 
}\label{FIG:two_plotsb}
\end{figure}
%%========================================================================================================================================================
%% THREE TABLES SHOWING RESULTS OF OPTIMIZED SIMULATION
\clearpage
\section{BBPT Versus FFMDPT Optimized Return Comparison} 
\begin{table}[htb]
\caption{Return Comparison of Bollinger Bands Pairs Trading
Simulation and Fixed Forecast Maximum Duration Pairs Trading
Simulation where $k = 1$ with $n$ optimized. \label{TBL:SIGONE}}
\medskip
\centering\begin{tabular}{lcrlcrr}
\hline \hline \hline
&  \multicolumn{2}{c}{BBPT STRATEGY} && \multicolumn{3}{c}{FFMDPT STRATEGY}  \\    
\cline{2-3} \cline{5-7}
Year& $n_\mathrm{BBPT}$  &  $\mathrm{RTN}_\mathrm{BBPT}$ 
   && $n_\mathrm{FFMDPT}$ & $\mathrm{RTN}_\mathrm{FFMDPT}$ & DIFF   \\
\hline 
2003-4  &  13 &  1.491      && 11 &  $-1.930$   &    3.4210   \\ 
2004-5  &  12 &  9.738      && 45 &  10.390     &  $-0.6529$  \\ 
2005-6  &  45 &  4.026      && 50 &  5.698      &  $-1.6720$  \\ 
2006-7  &  14 &  3.056      && 13 &  $-2.024$   &    5.0800    \\ 
2007-8  &  10 &  $-13.33$   && 20 &  $-3.464$   &  $-9.8660$  \\ 
2008-9  &  40 &  3.294      && 24 &  1.088      &    2.2060    \\  
2009-10 &  31 &  $-1.406$   && 28 &  $-3.625$   &    2.2190    \\ 
2010-11 &  18 &  $-0.7325$  && 10 &  1.810      &  $-2.5425$   \\
\hline\hline\hline
\end{tabular}
\end{table}

%%           BBPT           FFMDPT
%%        nstar return   nstar return        diff
%%2003     13    1.491     11   -1.930      3.4210
%%2004     12    9.738     45   10.390      -0.6520
%%2005     45    4.026     50   5.698       -1.6720
%%2006     14    3.056     13  -2.024       5.0800
%%2007     10   -13.33     10  -3.464       -9.8660
%%2008     40   3.294      24   1.088       2.2060
%%2009     31   -1.406     28  -3.625       2.2190
%%2010     18   -0.7325    10   1.810      -2.5425

%%=================================================================================================================================================

\begin{table}[htb]
\caption{Return Comparison of Bollinger Bands Pairs Trading
Simulation and Fixed Forecast Maximum Duration Pairs Trading
Simulation where $k = 1.5$ with $n$ optimized. \label{TBL:SIGONEPFIVE}}
\medskip
\centering\begin{tabular}{lcrlcrr}
\hline \hline \hline
&  \multicolumn{2}{c}{BBPT STRATEGY} && \multicolumn{3}{c}{FFMDPT STRATEGY}  \\    
\cline{2-3} \cline{5-7}
Year& $n_\mathrm{BBPT}$  &  $\mathrm{RTN}_\mathrm{BBPT}$ 
   && $n_\mathrm{FFMDPT}$ & $\mathrm{RTN}_\mathrm{FFMDPT}$ & DIFF   \\
\hline 
2003-4  &  13 &  3.104    && 12 &  $-3.898$   &   7.002    \\ 
2004-5  &  11 &  10.290   && 43 &    3.732    &   6.558    \\ 
2005-6  &  14 &  10.410   && 19 &    9.335    &   1.075     \\ 
2006-7  &  15 &   5.586   && 13 &    0.4854   &   5.1006     \\ 
2007-8  &  12 &  $-4.00$  && 10 &  $-2.331$   & $-1.6690$   \\ 
2008-9  &  15 &  $-3.154$ && 20 &    2.735    & $-5.8890$     \\  
2009-10 &  49 &  $-5.018$ && 50 &  $-9.310$   &   4.2920     \\ 
2010-11 &  16 &  0.0075   && 16 &  $-0.4846$  &   0.4921     \\
\hline\hline\hline
\end{tabular}
\end{table}

%%           BBPT           FFMDPT
%%        nstar return   nstar return   diff
%%2003     13 3.104       12   -3.898    7.002
%%2004     11 10.29       43    3.732    6.558
%%2005     14 10.41       19    9.335    1.075
%%2006     15 5.586       13    0.4854   5.1006
%%2007     12 -4.00       10   -2.331   -1.6690
%%2008     15 -3.154      20    2.735   -5.8890
%%2009     49 -5.018      50   -9.310    4.2920
%%2010     16  0.0075     16   -0.4846   0.4921

%%=================================================================================================================================================

\begin{table}[htb]
\caption{Return Comparison of Bollinger Bands Pairs Trading
Simulation and Fixed Forecast Maximum Duration Pairs Trading
Simulation where $k = 2$ with $n$ optimized. \label{TBL:SIGONEPTWO}}
\medskip
\centering\begin{tabular}{lcrlcrr}
\hline \hline \hline
&  \multicolumn{2}{c}{BBPT STRATEGY} && \multicolumn{3}{c}{FFMDPT STRATEGY}  \\    
\cline{2-3} \cline{5-7}
Year& $n_\mathrm{BBPT}$  &  $\mathrm{RTN}_\mathrm{BBPT}$ 
   && $n_\mathrm{FFMDPT}$ & $\mathrm{RTN}_\mathrm{FFMDPT}$ & DIFF   \\
\hline 
2003-4  &  32 &  2.162    && 11 &  2.989    &   $-0.8270$  \\ 
2004-5  &  11 &  4.115    && 38 &  2.097    &   2.0180     \\ 
2005-6  &  14 &  9.534    && 16 &  9.301    &   0.2330     \\ 
2006-7  &  14 &  4.728    && 15 &  0.770    &   3.9576     \\ 
2007-8  &  14 &  4.301    && 22 &  1.446    &   2.8550     \\ 
2008-9  &  15 &  0.548   && 11 &  $-2.851$ &    3.3988     \\  
2009-10 &  10 &  8.907    && 14 &  7.069    &   1.8380     \\ 
2010-11 &  14 &  0.802    && 12 &  1.456    &   $-0.6542$  \\
\hline\hline\hline
\end{tabular}
\end{table}

%%            BBPT           FFMDPT
%%        nstar return   nstar return    diff
%%2003    32    2.162     11    2.989    -0.8270
%%2004    11    4.115     38    2.097     2.0180
%%2005    14    9.534     16    9.301     0.2330
%%2006    14    4.728     15    0.7704    3.9576
%%2007    14    4.301     22    1.446     2.8550
%%2008    15    0.5478    11   -2.851     3.3988
%%2009    10    8.907     14    7.069     1.8380
%%2010    14    0.8018    12    1.456    -0.6542


\begin{thebibliography}{9}
\bibitem{JB01}
  John Bollinger,
 \emph{Bollinger On Bollinger Bonds},
   McGraw-Hill, 2002.

\bibitem{JB02}
  URL http://www.bollingerbands.com.

\bibitem{INV2007}
  http://www.investopedia.com/articles/trading/07/bollinger.asp

\bibitem{BK10}
 Butler M, Kazakov D,
 \emph{Particle Swarm Optimization of Bollinger Bands},
 Proceedings of the 7th international conference on Swarm
 intelligence, 2010

\bibitem{NZ2007}
 Ni M, Zhang C,
 \emph{An Efficient Implementation of the Backtesting of Trading Strategies},
 Springer-Verlag, Berlin, pp. 126-131 (2005).

\bibitem{OL10}
   Oleksiv,
   \emph{Statistical Analysis to the Optimization
   of the Technical Analysis Trading Tools: Trading Band
   Strategies}, 
   Ph.D. Thesis, 2010

\bibitem{TSC92}
Chande, Tushar S.,
\emph{Adapting Moving Averages to Market Volatility},
Technical Analysis of Stocks and Commodities, pp. 26-35 (March 1992)

\bibitem{DLT98}
Tilley, D.L., 
\emph{Moving Averages with Resistance and Support},
Technical Analysis of Stocks and Commodities, pp. 62-87 (Sep 1998)

\bibitem{FM73}
 Fama, E.F. and MacBeth, J.D. 
 \emph{Risk, Return and Equilibrium: Empirical Tests}
 Journal of Political Economy, pp, 607-636 (1973).

\bibitem{ZW03}
  Eric Zivot and Jiahui Wang,
  \emph{Modeling Financial Time Series with Splus},
   Springer-Verlag, New York, 2003.

\bibitem{AJ70}
  Jazwinski, A.,
 \emph{Stochastic Processes and Filtering Theory},
   Academic Press (1970).

\bibitem{AH92}
  Andrew C. Harvey,
 \emph{Forecasting structural time series models and the Kalman filter},
   Cambridge University Press, 1992.

\bibitem{JF2005}
  Julian Faraway,
 \emph{Linear Models with R},
   Chapman \& Hall/CRC, 2005.

\bibitem{R11}
  R Development Core Team (2011),
  \emph{R: A Language and Environment for Statistical Computing},
   R Foundation for Statistical Computing, Vienna Austria.
   ISBN 3-900051-0, URL http:://www.R-project.org/.

\bibitem{DH2011}
  Daniel Herlemont,
 \emph{Pairs Trading, Convergence Trading, Cointegration},
 URL http://www.yats.com/doc/cointegration-en.pdf.

\bibitem{RGB59}
  Robert Goodell Brown,
 \emph{Statistical forecasting for inventory control},
  McGraw-Hill, 1959.  

\bibitem{RGB63}
  Robert Goodell Brown,
 \emph{Smoothing, forecasting and prediction of discrete time series}
  Prentice Hall, EngleWood Cliffs, New Jersey, 1959.  

\bibitem{HKOS08}
  Hyndman, Koehler, Ord and Snyder,
  \emph{Forecasting with Exponential Smoothing The State Space Approach},
   Springer-Verlag, Berlin HeidelBerg, 2008.

\bibitem{BJR94}
  George Box, Gwilym Jenkins and Gregory Reinsel,
 \emph{Time Series Analysis, Forecasting and Control},
   Prentice Hall, Englewood Cliffs, N.J, 1994.

\bibitem{JM60}
  John F. Muth
 \emph{Optimal Properties of Exponentially Weighted Forecasts},
   Journal of the American Statistical Association,
   \textbf{55}, No. 290, (1960) 299--306.

\bibitem{CC00}
  Chris Chatfield,
 \emph{Time Series Forecasting},
   Chapman \& Hall/CRC, 2000.

\bibitem{ML2000}
  Mark Leeds,
 \emph{Error Structures for Dynamic Linear Models: Single Source versus Multiple Source}
   Ph.D. thesis, Pennsylvania State University, 2000.

\bibitem{KF60}
  Kalman, R.
 \emph{A New Approach to Linear Filtering and Prediction Problems}
   Transactions ASME Journal of Basic Engineering
   \textbf{82}, 35-45, (1960).

\bibitem{HS76}
  Harrison, P. and Stevens, C. ,
 \emph{Bayesian Forecasting (with discussion)},
   JRSSB,
   \textbf{38}, 205-247, (1976).

\bibitem{MO11}
  Andrew Moore,
  \emph{Cross Validation For Detecting and Preventing Overfitting}.
  http://www.autonlab.org/tutorials/overfit10.pdf

\end{thebibliography}
\end{document}